\documentclass[a4paper,allowtoday,onecolumn,accepted=2022-02-17]{quantumarticle} 
\pdfoutput=1 

\usepackage{dsfont} 	
\usepackage{microtype} 	
\usepackage[numbers,sort&compress]{natbib} 	
\usepackage[utf8]{inputenc} 	
\usepackage[T1]{fontenc} 	
\usepackage[british]{babel} 	
\usepackage[dvipsnames,x11names]{xcolor} 	
\usepackage{amsmath,amssymb,amsthm,bm,amsfonts,bbm} 	
\usepackage{mathtools} 	
\usepackage[colorlinks=true,citecolor=green,linkcolor=Purple,urlcolor=Magenta]{hyperref} 	

\newcommand{\ket}[1]{\vert#1\rangle}
\newcommand{\dket}[1]{\vert#1\rangle\hspace{-.8mm}\rangle}

\newcommand{\bra}[1]{\langle#1\vert}
\newcommand{\dbra}[1]{\langle\hspace{-.8mm}\langle #1\vert}

\newcommand{\ketbra}[2]{\vert #1 \rangle \hspace{-.4mm} \langle #2 \vert}
\newcommand{\dketbra}[2]{\vert #1 \rangle \hspace{-.8mm} \rangle \hspace{-.4mm} \langle\hspace{-.8mm}\langle #2 \vert}

\newcommand{\id}{\mathds{1}}

\DeclareMathOperator{\tr}{tr}
\newcommand*\dif{\mathop{}\!\mathrm{d}}
\newcommand{\mean}[1]{\left\langle#1\right\rangle}

\newcommand{\ie}{\textit{i.e.}}
\newcommand{\eg}{\textit{e.g.}}
\renewcommand{\H}{\mathcal{H}}

\renewcommand{\L}{\mathcal{L}}

\newcommand{\SU}{\mathcal{SU}}

\DeclarePairedDelimiter{\ceil}{\lceil}{\rceil}
\DeclarePairedDelimiter{\floor}{\lfloor}{\rfloor}
\makeatletter 
\newcommand{\doublewidetilde}[1]{{%
  \mathpalette\double@widetilde{#1}%
}}
\newcommand{\double@widetilde}[2]{%
  \sbox\z@{$\m@th#1\widetilde{#2}$}%
  \ht\z@=.9\ht\z@
  \widetilde{\box\z@}%
}
\newcommand{\map}[1]{\widetilde{#1}}  

\newcommand{\set}[1]{\mathsf{#1}}
\makeatother

\newtheorem{definition}{Definition}
\newtheorem{theorem}{Theorem}
\newtheorem{lemma}{Lemma}
\newtheorem{proposition}{Proposition}

\newcommand{\univie}{Faculty of Physics, University of Vienna, Boltzmanngasse 5, 1090 Vienna, Austria}
\newcommand{\IQOQI}{Institute for Quantum Optics and Quantum Information (IQOQI), Austrian
Academy of Sciences, Boltzmanngasse 3, A-1090 Vienna, Austria}

\hypersetup{pdftitle={{Deterministic Transformations Between Quantum Unitary Operations}}} 
\begin{document}
\title{Deterministic transformations between unitary operations: 
Exponential advantage with adaptive quantum circuits and the power of indefinite causality
}

\author{Marco T\'ulio Quintino}
\orcid{0000-0003-1332-3477}
\affiliation{\footnotesize\univie}
\affiliation{\IQOQI}
\author{Daniel Ebler}
\orcid{0000-0003-2696-8354}
\affiliation{Huawei Hong Kong Research Center, Hong Kong SAR, P. R. China}
\affiliation{Department of Computer Science, The University of Hong Kong, Pokfulam Road, Hong Kong}
\date{28th March 2022}

\begin{abstract}
	This work analyses the performance of quantum circuits and general processes to transform $k$ uses of an arbitrary unitary operation $U$ into another unitary operation $f(U)$. When the desired function $f$ a homomorphism, \ie, $f(UV)=f(U)f(V)$, it is known that optimal average fidelity is attainable by parallel circuits and indefinite causality does not provide any advantage. Here we show that the situation changes dramatically when considering anti-homomorphisms, \ie, $f(UV)=f(V)f(U)$. In particular, we prove that when $f$ is an anti-homomorphism, sequential circuits could exponentially outperform parallel ones and processes with indefinite causal order could outperform sequential ones. We presented explicit constructions on how to obtain such advantages for the unitary inversion task $f(U)=U^{-1}$ and the unitary transposition task $f(U)=U^T$. We also establish a one-to-one connection between three apparently different problems: unitary estimation, parallel unitary transposition, and parallel unitary inversion, allowing one to easily import results from one field to the other. Finally, we apply our results to several concrete problem instances and present a method based on computer-assisted proofs to show optimality.
\end{abstract}
\maketitle

\begin{figure}[h!] 
	\begin{center}
		\includegraphics[scale=0.45]{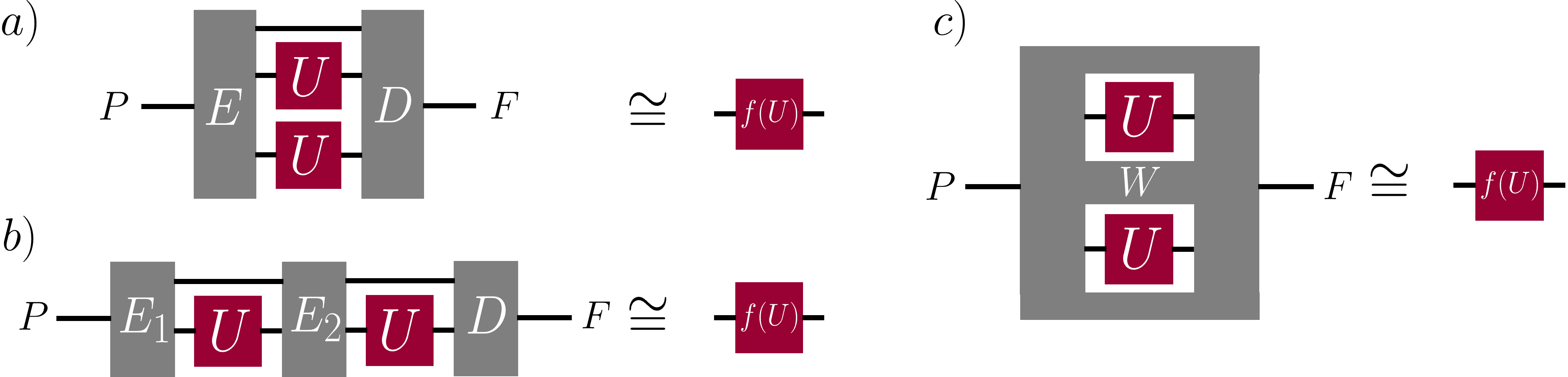} 
	\end{center}
\caption{Three different strategies to transform $k=2$ uses of an unknown unitary operation $U$ into $f(U)$. $a)$ parallel circuit, $E$ and $D$ stand for fixed operations (encoder and decoder); $b)$  sequential circuit with multiple encoder operations; and $c)$ general processes acting on $U$ may not have a definite causal order. Here, we analyse the performance of these strategies for different functions $f$. } \label{fig:problem}
\end{figure}

\clearpage
\tableofcontents
\newpage

\section{Introduction}
    Classical computer science is based upon the concept of perfect control. The states, described by discrete bit values 0 or 1, can be prepared, copied and read off at will. Manipulations of input bit strings succeed by executing operations with functional descriptions $h$. The same holds true for the transformations thereof: it is natural to take a function $h$ -- which can be of arbitrary complexity -- and utilise it as a subroutine in a larger code $f$, leading to an assignment $f(h)$. Notably, the body $h$ of such code $f$ can be partially or even fully unknown~\cite{bird1988functional}. For instance, conditional statements, such as if-clauses, are usually defined independently of the body, leading to large modularity of implementations.

    These seemingly straightforward principles have proven highly problematic for computation in the quantum domain. Early studies discovered limitations on state manipulation such as the impossibility of copying quantum states~\cite{wootters82}, the impossibility of constructing universal not gates~\cite{buzek99}, and tradeoff relations between information gain and state disturbance~\cite{fuchs98information}. More recently, restrictions on manipulating operations have also been identified, such as the  the impossibility of gate replication \cite{chiribella08clone}, the impossibility of an if-clause for unknown operations \cite{araujo14control,bisio15control,dong19control,gavorova20control}, and the impossibility of iterating unknown gates \cite{soleimanifar16}.
    
    A natural path around no-go theorems of quantum information processing is to utilise multiple instances of the input resources. For instance, a quantum state can be cloned with a higher fidelity rate when multiple copies are available~\cite{werner98,bruss98clone} and positive but not completely positive operations such as the universal not gate can be better approximated by consuming multiple copies~\cite{buzek99,dong18}. When transforming an unknown unitary $U$ in accordance with a function $f$, multiple copies, or simply, multiple calls to the same unitary $U$, may be utilised to perform transformations $f(U)$ with higher fidelity than a single use of $U$. This introduces the freedom to arrange the multiple uses of $U$ in different configurations. One could, for instance, apply all uses in parallel or concatenate them in a sequence such that the output of one use serves as the input of the next. From a more general perspective, one could even consider a process where the operations are used without respecting a definite causal order~\cite{chiribella09_switch,oreshkov11}.
    
     A protocol realizing a function $f(U)$ with multiple uses of the operation $U$ is usually categorised as either deterministic or probabilistic exact: the former describes the situation when the output is always accepted, yielding usually approximations to $f(U)$, while the latter implements the transformation exactly with a certain probability of success, maintaining the option to discard failed attempts. The tasks of deterministic unitary cloning~\cite{chiribella08clone}, deterministic unitary complex conjugation~\cite{miyazaki17,ebler16}, and deterministic unitary inversion~\cite{ebler16} were analysed in detail for the particular scenario where a single use of the input operation is available. When considering multiple uses, optimal cloning of unitary operations~\cite{chiribella08clone}, estimating~\cite{chiribella05estimation,chiribella08memory_effects} and learning~\cite{bisio10} unitary operations, and transformations between different representations of the same group~\cite{bisio13} can always be performed in parallel -- suggesting that sequential strategies may not lead to a substantial improvement in performance. In the probabilistic scenario, sequential circuits are known to exponentially outperform parallel ones for unitary inversion~\cite{quintino19PRL,feng20inversion} and unitary transposition~\cite{quintino19PRA}. However, up to our work, the true potential of sequential deterministic strategies remained unexplored.

    In this work, we develop further insights regarding deterministic transformations of unknown unitary actions $U$ into operations $f(U)$. Qualitatively, our results include novel understandings about an exponential performance gap between parallel and sequential implementations for certain classes of operations. This complements previous investigations, which found that parallel implementations of $U$ are optimal for the restricted class of functions $f$ mapping representations of a compact group $G$ to other representations of the same group \cite{bisio13}. In particular, Ref.~\cite{bisio13} considers functions $f$ which are a homomorphism, \ie, $f(UV)=f(U)f(V)$ for every unitary operator $U$ and $V$. In order to overcome this assumption, we analyse functions which are anti-homomorphism,\ie, $f(UV)=f(V)f(U)$. Particularly, for unitary inversion ($f(U)= U^{-1}$) and unitary transposition ($f(U)=U^T$) we show that the hypothesis of Ref.~\cite{bisio13} do not apply, and sequential strategies provide an exponential improvement in performance. Finally, we investigate whether indefinite causal structures can further boost the performance of implementations. Despite the fact that the popular quantum switch \cite{chiribella09_switch} and its generalisations cannot outperform causally ordered strategies when $k$ uses of the same unitary are considered~\cite{bavaresco21b}, we show that other indefinite causal structures indeed yield a strict advantage over any causally ordered implementation.


\section{Notation and mathematical preliminaries}
			In this section, we revise the mathematical framework to analyse transformations between quantum operators. This formalism, also referred as higher-order quantum operations~\cite{perinotti16higher,bisio19higher}, was developed in various references with different perspectives and motivations~\cite{kretschmann05,chiribella07,gutoski07,zyczkowski08quartic,chiribella09networks,pollock15markovian}. Also, its mathematical formulation have proven to be rich enough to provide novel insights to the foundations of quantum theory and general probabilistic theories \cite{chiribella11informational,chiribella09purification},  causality~\cite{costa16causal,ried15causal,feix17causal,nery21CCDC}, indefinite causality~\cite{oreshkov11,araujo16}, quantum network design~\cite{chiribella2012optimal,ebler16}, and quantum stochastic processes~\cite{milz21markov}.
	
	\subsection{Basic quantum elements and the Choi isomorphism}
	We now list basic definitions and establish the notation which will be used along the paper.	
\begin{itemize}
\item $\H$ stands for complex linear (Hilbert) spaces of finite dimension, \ie, $\H\cong\mathbb{C}_d$
\item The \textit{Choi vector} $\dket{U}\in \H_\text{in}\otimes\H_\text{out}$ of a linear operator $U:\H_\text{in}\to\H_\text{out}$ is defined as
\begin{equation}
	\dket{U}:=\sum_i \ket{i}_\text{in}\otimes (U\ket{i})_\text{out} 
\end{equation}
where $\{\ket{i}\}_i$ is the computational basis
\item $\mathcal{L}(\H)$ stands for the set of linear operators acting on $\H$.
Linear transformations between operators are referred to as \textit{linear maps} and marked with an upper tilde, \eg, ${\map{C}:\mathcal{L}(\H_\text{in})\to\mathcal{L}(\H_\text{out})}$
\item The \textit{Choi operator} $C\in \mathcal{L}(\H_\text{in}\otimes\H_\text{out})$ of $\map{C}$ of a linear map $\map{C}:\mathcal{L}(\H_\text{in})\to\mathcal{L}(\H_\text{out})$ is defined as
\begin{equation}
	C:=\sum_{ij}\ketbra{i}{j}\otimes \map{C}(\ketbra{i}{j}).
\end{equation}
\end{itemize}
	
	Quantum states represent our knowledge of quantum systems, and deterministic transformations between quantum states are described by quantum channels. 
	 
\begin{itemize}
\item A quantum state is a linear operators $\rho\in\L(\H)$ which is positive semidefinite, $\rho\geq0$, and has unity trace, $\tr(\rho)=1$
\item  A quantum channel is a linear map 
$\map{C}:\L(\H_\text{in}) \to\L(\H_\text{out})$ which are completely positive and trace preserving. In Choi operator notation these constraints correspond to
\begin{align}
C\geq0  &\iff \map{C} \text{ is completely positive} \nonumber \\
\tr_\text{out}(C)=\id_\text{in} &\iff \map{C} \text{ is trace preserving}
\end{align}
\item A quantum channel $\map{C}$ is unitary if it can be written as $\map{C}(\rho)=U\rho U^\dagger$ where $U:\H_\text{in}\to\H_\text{out}$ is a unitary operator, \eg, $UU^\dagger=U^\dagger U=\id$. The Choi of a unitary channel $\map{C}(\rho)=U\rho U^\dagger$ can be written as $C=\dketbra{U}{U}\in\L(\H_\text{in}\otimes\H_\text{out})$
\end{itemize}
	Unitary quantum channels form a very important class of transformations, they represent reversible quantum transformations, describe dynamics in closed quantum systems and quantum gates.

	When dealing with Choi operators, composition of linear maps can be conveniently expressed in terms of the \textit{link product} \cite{chiribella07}. In particular, with the link product, denoted by $*$, we can write
\begin{alignat}{2}
&\textbf{Map: } \map{A}:\L(\H_1)\to\L(\H_2) \quad\quad\quad\quad &&\textbf{Choi: } A\in\mathcal{L}(\H_1\otimes\H_2)\nonumber \\ 
&\textbf{Map: } \map{B}:\L(\H_2)\to\L(\H_3) \quad\quad\quad\quad &&\textbf{Choi: } B\in\mathcal{L}(\H_2\otimes\H_3) \nonumber \\
&\textbf{Map: } \map{C}:=\map{B}\circ \map{A} \quad\quad\quad\quad && \textbf{Choi: } C=B*A, 
\end{alignat}		
where the link product $A*B$ is defined as 
\begin{align}
	A*B:=&\tr_2\left( [A^{T_2} \otimes \id_3] \,[ \id_1 \otimes B ]\right) 
\end{align}
with $A^{T_2}$ being the partial transposition on the linear space $\H_2$. If we keep track on the linear spaces where the operators act, the link product is commutative, $A*B=B*A$, and associative, $A*(B*C)=(A*B)*C$. As shown in the next subsections, these properties will be very useful to represent ``incomplete'' quantum circuits and quantum superchannels. 

\subsection{1-slot quantum superchannels}	\label{sec:1slot}
\begin{figure}[h!] 
	\begin{center}
		\includegraphics[scale=0.55]{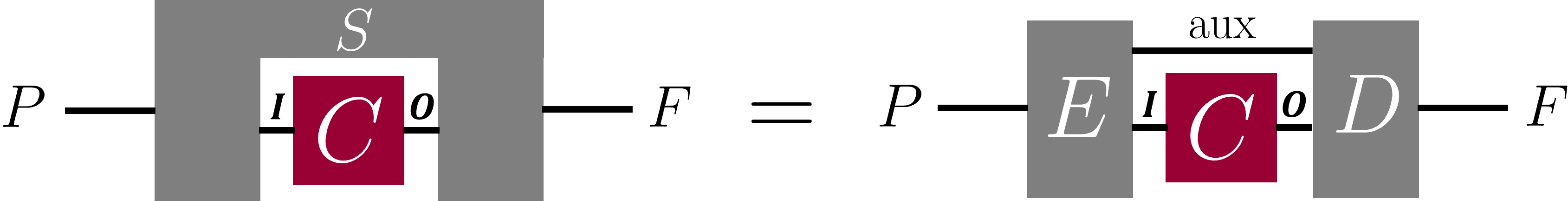} 
	\end{center}
\caption{Circuit illustration of a 1-slot quantum superchannel. Quantum operations from the input space $I$ to the output space $O$ with Choi operator $C\in\L(\H_{{I}}\otimes\H_{{O}})$ may be ``plugged'' into the slot to obtain another quantum operation from past space $P$ to future space $F$ channel described by $S*C\in\L(\H_P\otimes\H_F)$. Every quantum superchannel can be realised by composing an encoder quantum channel $E\in\L(\H_P\otimes\H_{{I}}\otimes\H_\text{aux})$ and a decoder quantum channel $D\in\L(\H_{O}\otimes\H_\text{aux}\otimes\H_F)$ such that $S=E*D$. } \label{fig:1_slot}
\end{figure}
	
	We now review the concept of a $1$-slot \textit{quantum superchannel}\footnote{One slot quantum superchannels are also known under the name of $1$-slot quantum combs \cite{chiribella07} and are equivalent to $2$-turn quantum strategies \cite{gutoski07} and a bipartite channel with memory \cite{kretschmann05}.}, a quantum process that transforms quantum channels into quantum channels \cite{chiribella08}. Before presenting formal definitions, let us illustrate the idea of a superchannel by considering a quantum circuit composed by concatenation of three quantum channels described as following (see Fig.~\ref{fig:1_slot}). 
\begin{enumerate}
\item We start a quantum circuit with a quantum channel $\map{E}$ which we refer as encoder. The encoder channel transforms states from a ``past'' system, denoted as $\H_P$ to a composed, $\H_{{I}}\otimes\H_\text{aux}$ where $\H_{{I}}$ stands for input space and $\H_\text{aux}$ auxiliary. 

At this stage, the circuit is described by $\map{E}$.
\item Secondly, we introduce an arbitrary quantum channel $\map{C}$ which transforms states from the input space $\H_{{I}}$ to an output space $\H_{{O}}$. At this point, the auxiliary space $\H_\text{aux}$ stays untouched. 

At this stage the circuit is described by $\left(\map{\id}_\text{aux} \otimes \map{C}\right) \circ \map{E}$
\item Finally, we perform quantum channel $\map{D}$ which we refer as decoder. The decoder channel transforms states from the composed space $\H_{{O}}\otimes\H_\text{aux}$ to a ``future'' space $\H_F$.

At this stage the circuit is described by $\map{D}\circ\left(\map{\id}_\text{aux} \otimes \map{C}\right) \circ \map{E}$
\end{enumerate}
The mathematical description of the circuit elements can be summarised by:
\begin{alignat}{2}
&\text{\textbf{Map:} } \map{E}:\L(\H_P)\to\L(\H_{I}\otimes\H_\text{aux}) \quad\quad\quad\quad\quad &&\textbf{Choi: } E\in\mathcal{L}(\H_P\otimes\H_{I}\otimes\H_\text{aux}) \nonumber \\ 
&\textbf{Map: } \map{C}:\L(\H_{I})\to\L(\H_{O}) \quad\quad\quad\quad\quad &&\textbf{Choi: } C\in\mathcal{L}(\H_{I}\otimes\H_{O}) \nonumber \\
&\textbf{Map: } \map{D}:\L(\H_\text{aux}\otimes\H_{O})\to\L(\H_F) \quad\quad\quad\quad\quad &&\textbf{Choi: } D\in\mathcal{L}(\H_\text{aux}\otimes\H_{I}\otimes\H_F) \nonumber \\ 
&\textbf{Map: } \map{D}\circ \left(\map{\id}_\text{aux} \otimes \map{C}\right) \circ \map{E} \quad\quad\quad && \textbf{Choi: } D*C*E, 
\end{alignat}	
where $\map{\id}$ is the identity map, \textit{i.e.,} $\map{\id}(\rho)=\rho$ for any linear operator $\rho$. We now point out a convenient property of the link product. When composing quantum operations with the link product, the action of the identity map $\map{\id}$ can be neglected. This property follows from the mathematical identity, 
\begin{equation}
	D*\left(\dketbra{\id}{\id}_\text{aux,aux}\otimes C_{IO}\right) *E=D*C*E,
\end{equation}
where $\dketbra{\id}{\id}$ is the Choi operator of the identity map $\map{\id}$. Next, we are going to exploit the associativity and commutativity of the link product to conveniently represent ``incomplete quantum circuits''.

	Consider a situation where the quantum channel $\map{C}$ is not plugged into the circuit. We then have an ``incomplete circuit'' which only has an encoder and a decoder operation, with open input and output wires. Thanks to the associativity and commutativity of the link product, we can represent this incomplete circuit as $S:=E*D$. We will soon link $S$ to the concept of superchannels. In this way, for any quantum channel $C$ which we plug in this circuit, the output operation is described by, 
\begin{align}
	S*C=&E*D*C \nonumber \\
	 =&D*C*E, \label{eq:1slotdec}
\end{align}
which is a quantum channel from the past space $\H_P$ to the future space $\H_F$.

\begin{definition}[1-slot superchannel \label{def:1slot}]
	A linear operator $S\in \L(\H_P\otimes \H_{I} \otimes \H_{O} \otimes \H_F)$ is a (1-slot) superchannel if there exist a linear space \emph{$\H_\text{aux}$}, quantum channels \emph{$\map{E}:\L(\H_P) \to \L(H_\text{aux}\otimes \H_{I})$} and \emph{$\map{D}:\L(\H_\text{aux}\otimes \H_{O})\to\L(\H_F)$} such that $S=E*D$.
\end{definition}

 A 1-slot superchannel $S$ may be seen physically as a generalised circuit board with one input slot. For any channel $C$ inserted, $S$ will function the same way: it pre-processes an input state with $E$, forwards (part of) the output to $C$ and processes afterwards with $D$.  From a more abstract perspective, quantum superchannels may be seen as objects which transform quantum operations into quantum operations. Although we have motivated and presented the definition of quantum superchannels in terms of quantum circuits, 1-slot quantum superchannels can be equivalently defined as the most general ``quantum object'' which transforms arbitrary channels into arbitrary channels \cite{chiribella08}. This equivalence may be viewed as a higher-order generalisation of Stinespring dilation \cite{chiribella09networks}.

	Quantum superchannels also admit a characterisation in terms of linear and positive definite constraints. A linear operator $S\in \L(\H_P\otimes \H_{I} \otimes \H_{O} \otimes \H_F)$ is a 1-slot superchannel if and only if it respects \cite{chiribella07,chiribella08,chiribella09networks}
\begin{align}\label{eq:1slot}
	S &\geq 0 			\nonumber \\
  	\tr_F(S)  &= \tr_{{O}F}(S)\otimes \frac{\id_{O}}{d_{O}} 		\nonumber \\
  	\tr_{IOF}(S)&= \tr_{PIOF}(S)\otimes \frac{\id_P}{d_P} \nonumber \\
 	\tr(S) &= d_P d_{{O}}.
\end{align}	

\subsection{Multi-slot superchannels} \label{sec:k-slots}
	In this section, we extend the notion of superchannel from a single input channel ${\map{C}:\L(\H_{\bm{I}})\to\L(\H_{\bm{O}})}$ to transformations acting on a set of $k$ channels with $ \{ \map{C_i}\}_{i=1}^k$, $\map{C_i}:\L(\H_{I_i})\to\L(\H_{O_i})$. For the multi-slot case, the linear spaces associated to input and output are described by the tensor product of $i$ subspaces. In this work, we use bold letters to indicate this tensor product subsystem structure, that is, 
	$\H_{\bm{I}}:=\bigotimes_{i=1}^k \H_{I_i}$ and $\H_{\bm{O}}:=\bigotimes_{i=1}^k \H_{O_i}$. Differently from the single slot case,  there are multiple ways to position the $k$ input-channels.
	In particular, input channels may be placed in parallel, in sequence (as in an adaptive protocol), or in a general manner where order of input channels may be indefinite, \textit{e.g.} via the quantum switch \cite{chiribella09_switch} or general process matrices \cite{oreshkov11,araujo15}. We refer to Fig.~\ref{fig:superchannels} for a pictorial illustration of these different classes of superchannels for the case of  $k=2$ slots.
	
\begin{figure}[h!] 
	\begin{center}
		\includegraphics[scale=0.46]{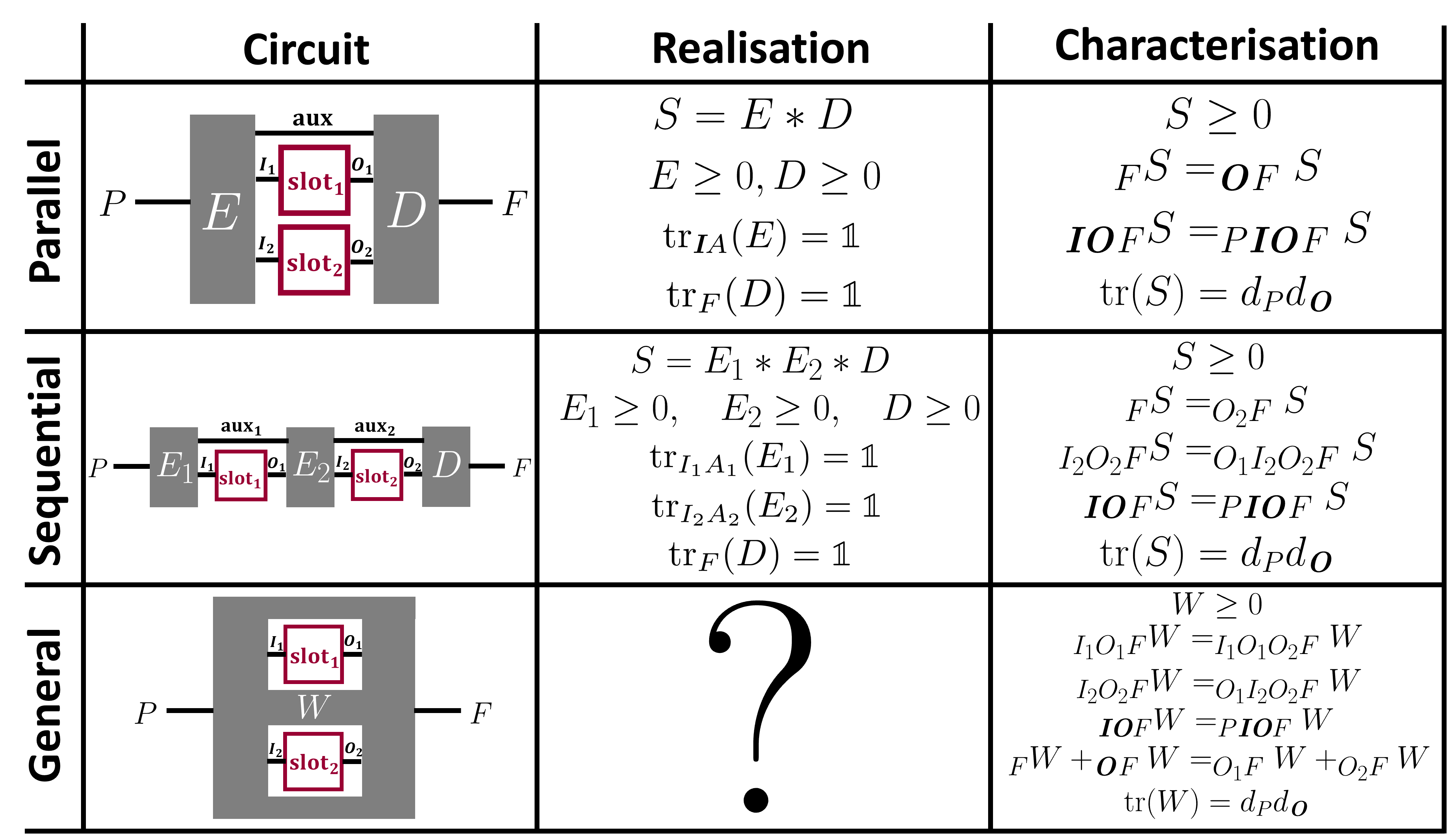} 
	\end{center}
\caption{Table summarising the definition and characterisation of parallel, sequential, and general $k=2$ slots superchannels. Some equations make use of the trace-and-replace notation $_iS:=\tr_i(S)\otimes\frac{\id_i}{d_i}$. Note that since the dimension of auxiliary spaces are not restricted, these three classes form a hierarchy of superchannels (see Fig.~\ref{fig:seq>par} for an illustration).}
\label{fig:superchannels}
\end{figure}

\subsubsection{Parallel superchannels}
	Parallel superchannels can be characterised by a single encoder and a single decoder channel. To make the analogy to the single-slot case, the input channel $C$ in Eq.~(\ref{eq:1slotdec}) is replaced by the tensor product of the $k$ input channels  $\{\map{C_i}\}_{i=1}^k$  as $\map{C} :=\bigotimes_{i=1}^k \map{C_i}$. This allows us to treat parallel superchannels in the single-slot setting described in definition \ref{def:1slot}. 

\begin{definition}[Parallel superchannel]
	A linear operator $S\in \L(\H_P\otimes \H_{\bm{I}} \otimes \H_{\bm{O}} \otimes \H_F)$ is a $k$-slot parallel superchannel if there exist a linear space \emph{$\H_\text{aux}$}, a quantum channel \emph{$\map{E}:\L(\H_P) \to \L(H_\text{aux}\otimes \H_{\bm{I}})$}, and \emph{$\map{D}:\L(\H_\text{aux}\otimes \H_{\bm{O}})\to\L(\H_F)$} with Choi operators $E$ and $D$ such that $S=E*D$.
\end{definition}
	As in the single slot case, it can be shown that a linear operator $S\in \L(\H_P\otimes \H_{\bm{I}} \otimes \H_{\bm{O}} \otimes \H_F)$ is a $k$-slot parallel superchannel if and only if
\begin{align} 
	S &\geq 0 			\nonumber \\
  	\tr_F(S)  &= \tr_{{\bm{O}}F}(S)\otimes \frac{\id_{\bm{O}}}{d_{\bm{O}}} 		\nonumber \\
  	\tr_{\bm{IO}F}(S)&= \tr_{P\bm{IO}F}(S)\otimes \frac{\id_P}{d_P} \nonumber \\
 	\tr(S) &= d_P d_{{\bm{O}}}. \label{eq:parcond}
\end{align}	

	When transforming quantum operations, parallel implementations are often desirable due to their simpler structure, they can be realised by a single encoder and a single decoder channel. Also, parallel superchannels can be realised by a quantum circuit with short depth (encoder, input-channels, decoder) while a sequential use of the input operations may result in a long depth, and consequently, in a longer time to finish the whole transformation. Especially for current generation hardware implementations prone to short coherence times, long computing times deem challenging.

\subsubsection{Sequential superchannels}
	Sequential superchannels%
\footnote{We remark that $k$-slots superchannels are also known in the literature as $k$-slot combs \cite{chiribella07,chiribella09networks} which are equivalent to $(k+1)$-turn quantum strategies \cite{gutoski07}, $(k+1)$-partite channel with memory \cite{kretschmann05}, and $k$-partite ordered process matrices with common past and common future \cite{araujo16}.} %
	represent general quantum circuits where different encoders operations are applied in between the use of the input channels $\map{C_i}$. For instance, in the case of $k=2$ slots, sequential superchannels consist of two encoding channels with Choi operators ${E_1}$, ${E_2}$ and one decoder channel with Choi operator ${D}$ (see Fig.~\ref{fig:superchannels} for an illustration of this case). If we plug in two input channels with Choi operators ${C_1}$ and ${C_2}$, the output channel $C_\text{out}$ is given by the composition $C_\text{out}=D*C_2*E_2*C_1*E_1$. Formally, we can define sequential superchannels as follows.

\begin{definition}[Sequential superchannel]
	A linear operator $S\in \L(\H_P\otimes \H_{\bm{I}} \otimes \H_{\bm{O}} \otimes \H_F)$ is a $k$-slot sequential superchannel if there exist a linear space \emph{$\H_\text{aux}$}, a quantum channel \emph{$\map{E_1}:\L(\H_P) \to \L(H_\text{aux}\otimes \H_{I_1})$},
a set of quantum channels 
$\map{E_i}: \L(H_\text{aux} \otimes \H_{O_{i-1}})$ $\to$ $\L(H_\text{aux}\otimes \H_{I_i})$
for $i\in \{2,\ldots,k\}$, and a quantum channel
${ \map{D}: \L(\H_\text{\emph{aux}} \otimes \H_{O_k}) \to\L(\H_F)}$ such that
\begin{equation}
	S=E_1*E_2*\ldots *E_k*D.
\end{equation}
\end{definition}
	Sequential superchannels can also be characterised in terms of linear and positive semidefinite constraints. A linear operator $S\in \L(\H_P\otimes \H_{\bm{I}} \otimes \H_{\bm{O}} \otimes \H_F)$ represents a sequential superchannel with $k$-slots if and only if \cite{chiribella07,chiribella09networks}
		 \begin{align}
		 	S &\geq 0 \nonumber \\
		 	\tr_{F} (S) & = \tr_{O_kF}(S) \otimes \frac{\id_{O_k}}{d_{O_k}} \nonumber \\ 
		 	\tr_{I_kO_kF}(S) & = \tr_{O_{k-1}I_kO_kF}(S) \otimes \frac{\id_{O_{k-1}}}{d_{O_{k-1}}} \nonumber \\ 
		 	&\;\, \vdots \nonumber \\
		 	\tr_{I_1O_1\ldots I_kO_kF}(S) & = \tr_{PI_1O_1\ldots I_kO_kF}(S)\otimes \frac{\id_{P}}{d_{P}}\nonumber \\ 
		 	\tr(S)& = d_P d_{\bm{O}}. \label{eq:seqcond}
		 \end{align}

	We remark that since the size of the auxiliary space $\H_\text{aux}$ is not restricted, parallel superchannels may be seen as a particular instance of sequential superchannels. That is, any parallel superchannel given by $S'=E*D_\text{par}$ can be written as a sequential superchannel $S_\text{seq}=E_1*E_2*\ldots*E_k*D_\text{seq}$ by setting $E_1:=E$, $D_\text{seq}:=D$ and the other encoder operations $E_2,\ldots,E_k$ as swap operations. This construction is illustrated in Fig.~\ref{fig:seq>par}.
	 
\begin{figure}[h!] 
	\begin{center}
		\includegraphics[scale=0.55]{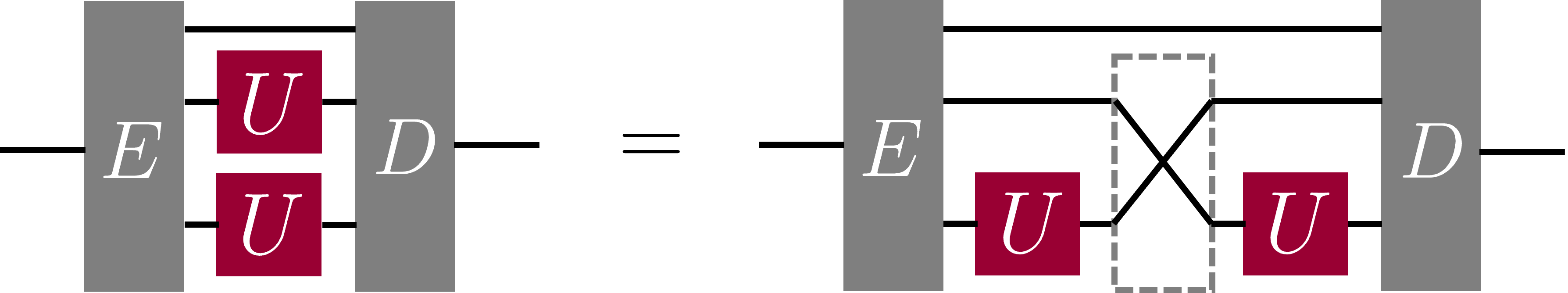} 
	\end{center}
\caption{Pictorial illustration on how one can write every $k=2$-slot parallel superchannel $S'=E*D$ as $k=2$-slot sequential one $S_\text{seq}=E_1*E_2*D$.}
\label{fig:seq>par}
\end{figure}

\subsubsection{General superchannels}
	We now describe the most general class of $k$-slot superchannels, encompassing any possible map compatible with quantum theory. The classes of parallel and sequential superchannels are a subset of this general class -- corresponding to definite orders of input channels -- complemented by arrangements of input slots for which the order cannot be specified any more. The latter class of superchannels describes the structural background of indefinite causal structures, with the quantum switch \cite{chiribella09_switch} being one prominent instance in which sequential circuits with different causal orders are  coherently superposed (for physical realisations, see for example \cite{rubino17,goswami18,goswami20,rubino21}).

	More concretely, a general superchannel is the most general linear function that transforms $k$ independent quantum channels into a quantum channel, even when this linear function is applied only on part of these input channels.  This is made precise in the following definition.

\begin{definition}[General superchannel] \label{def:general}
	A linear operator $S\in \L(\H_P\otimes \H_{\bm{I}} \otimes \H_{\bm{O}} \otimes \H_F)$ is a $k$-slot general superchannel if $S\geq0$ and, for every possible set of quantum channels $\{\map{C_i}\}_{i=1}^k$, $\map{C_i}:\H_{I_i}\to\H_{O_i}$ we have that $S* \left( C_1\otimes C_2 \otimes \ldots\otimes C_k\right)$ is a quantum channel.
\end{definition}
	Similarly to the parallel and sequential case, general superchannels also have a characterisation in terms of linear and positive semi-definite constraints. A general characterisation in terms of dual affine sets is presented in \cite{ebler16}. Trace conditions similar to Eqs. (\ref{eq:parcond}, \ref{eq:seqcond}) for parallel and sequential superchannels can be derived with  for any number of slots $k$ with methods given in \cite{araujo16}. 
	
	When considering $k=1$ slots, general quantum superchannels correspond to the most general ``quantum transformation'' which maps quantum channels into quantum channels. Interestingly, Ref.~\cite{chiribella08} shows that every $k=1$ slot superchannel can be realised by a single-slot superchannel as in Def.~\ref{def:1slot}. Hence, the constructive definition of single-slot superchannel based on an encoder and decoder channel presented in Section~\ref{sec:1slot} is equivalent to stating that single-slot superchannels are general single-slot superchannels.
	
	For the case $k=2$, an operator $S\in\L(\H_P\otimes\H_{\bm{I}}\otimes\H_{\bm{O}}\otimes\H_F)$ is a general superchannel if and only if \cite{araujo16}
\begin{align}
   S&\geq 0   \nonumber \\
    _{I_1O_1F}S &= _{I_1O_1O_2F}S \nonumber \\
    _{I_2O_2F}S &= _{O_1I_2O_2F}S  \nonumber \\
    _FS  &= _{O_1F}S + _{O_2F}S - _{O_1O_2F}S  \nonumber \\
    _{\bm{IO}F}S  &= _{P\bm{IO}F}S  \nonumber \\
    \tr S&=d_Pd_{O_1}d_{O_2}  .
\end{align}
Above, we utilised the trace-and-replace notation 
	\begin{equation} \label{eq:TR}
		_iS:=\tr_i(S)\otimes\frac{\id_i}{d_i}
	\end{equation} 
	to denote tracing the system $i$ and replacing it by the normalised identity map on the system $\H_i$. An explicit characterisation for general channels with $k=3$ slots is presented in \cite{quintino19PRA}. 
	
	Currently, it is not know whether there exists a physical procedure which is able to realise arbitrary general processes in a ``fair'' manner. We note however that relevant and non-trivial classes of general superchannels with indefinite causal order may be implemented by means of coherent control of sequential superchannels \cite{wechs21control}. Also, from a theoretical perspective, Ref.~\cite{araujo16} proposes a purification postulate in which only superchannels which preserve reversibility may have fair physical realisation. For the two-slot case, non-trivial reversibility preserving superchannels are proven to be switch-like superchannels~\cite{yokojima20,barrett20cyclic}, and as proven in Ref.~\cite{bavaresco21b} and discussed in Sec.~\ref{sec:advantage_indefinite}, switch-like superchannels cannot outperform sequential superchannels for transforming multiple copies of the same unitary operation.

\subsubsection{Parallel measure-and-prepare superchannels} \label{sec:measure-and-prep_and_delayed}
	Parallel measure-and-prepare strategies correspond to superchannels which measure the input operation and -- conditional on the measurement outcome $i$ -- prepare the channel $\map{R_i}$ (see Fig.~\ref{fig:prepare_and_measure_delayed}).  	Measure-and-prepare strategies are closely related to estimation of quantum channels (see Sec.~\ref{sec:trans=est} for more detailed discussion), since it consists of guessing the input operation and then preparing an output operation based on the guess. This class of strategies may be viewed as a ``most classical" manner to transform a quantum operation, which is of particular interest for experimental implementations.
	
	A parallel measure-and-prepare superchannels may be described as the following:
\begin{enumerate}
\item Prepare a quantum probe-state $\rho\in\L(\H_{\bm{I}}\otimes\H_\text{aux})$.
\item The input-channel $\map{C}:\L(\H_{\bm{I}})\to\L(\H_{\bm{O}})$  is performed in the input space $\H_{\bm{I}}$ of the probe-state $\rho\in\L(\H_{\bm{I}}\otimes\H_\text{aux})$  to obtain $\map{C}\otimes\map{\id}(\rho)=C*\rho$ .
\item A quantum measurement with POVM elements $M_i\in \L(\H_{\bm{O}}\otimes\H_\text{aux})$ is performed on the state $\map{C}\otimes\map{\id}(\rho)$ to obtain the outcome $i$ with probability 
	$\tr\Big((C*\rho)M_i \Big) = C*\rho*M_i^T$. When the result $i$ is obtained, we prepare the output channel $\map{R_i}:\L(\H_P)\to\L(\H_F)$. On average, the output channel is described by\footnote{In principle, we may also consider measurements with infinitely many outcomes, in this case the sum should be described as an integral. We remark however that the extremal  measurements of $d$-dimensional systems have at most $d^2$ outcomes~\cite{dariano05,chiribella07finitePOVM}. Due to that, in some optimisation problems, like the ones discussed in this paper, it is enough to consider measurements with finite outcomes.}
$\sum_i C\,*\rho*M_i^T * R_i $.
\end{enumerate}
	Note that an of operators $\{M_i^T\}_i$ is a valid POVM if and only if the set $\{M_i\}_i$ is a valid POVM. Thanks to that, with no loss of generality, we may define a general parallel measure-and-prepare superchannel without explicitly writing the transposition on the measurement operators.

\begin{definition}[Parallel measure-and-prepare superchannel]
A linear operator $S\in \L(\H_P\otimes \H_{\bm{I}} \otimes \H_{\bm{O}} \otimes \H_F)$ is a $k$-slot measure-and-prepare superchannel if there exist a linear space \emph{$\H_\text{aux}$}, a quantum state \emph{$\rho\in L(\H_{\bm{I}}\otimes \H_\text{aux})$}, quantum measurement with POVM elements \emph{$M_i\in\L(\H_{\bm{O}}\otimes\H_\text{aux})$}, and a set of quantum channels \emph{$\map{R_i}:\L(\H_P)\to\L(\H_F)$} such that $S=\sum_i\rho*M_i*R_i$.
\end{definition}
\begin{figure}[h!] 
	\begin{center}
		 \includegraphics[width=0.95\columnwidth]{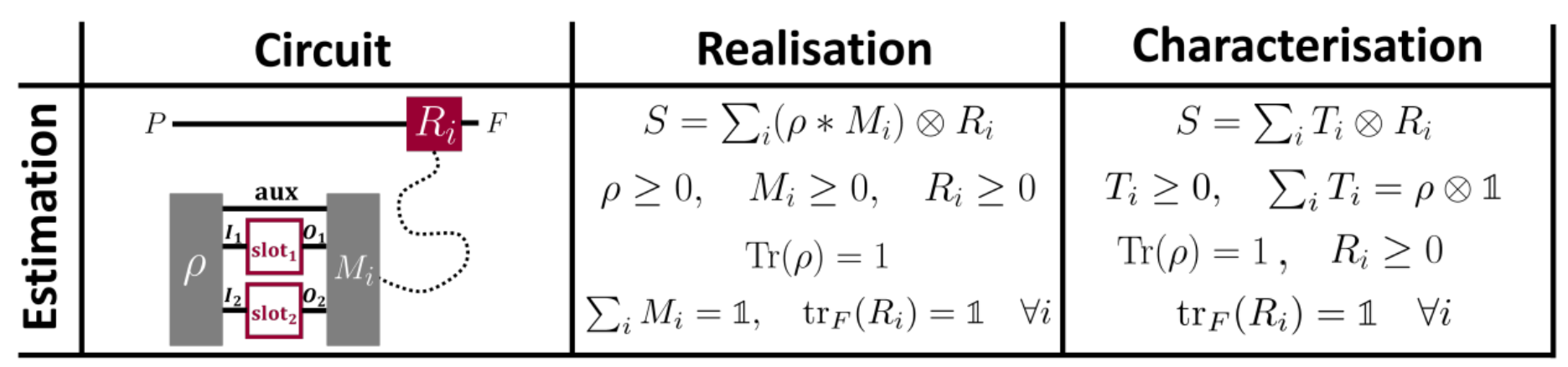} 
	\end{center}
\caption{Pictorial illustration of a measure and prepare and a delayed-input-state parallel superchannel.}
\label{fig:prepare_and_measure_delayed}
\end{figure}

	Parallel measure-and-prepare superchannels are strategies in which we measure a quantum channel, and prepare an output channel accordingly to the measurement outcome. As described earlier, we can measure a quantum channel by preparing a quantum state, sending part of it though the channel, and then performing a quantum measurement. We can nevertheless simplify this procedure by means of a (parallel) quantum tester \cite{chiribella09networks,ziman08PPOVM,bavaresco21}. A (parallel) quantum tester is a set of positive semidefinite operators $\{T_i\}_i$, $T_i\in\L(\H_{\bm{I}}\otimes\H_{\bm{O}})$ such that $\sum_i T_i=\sigma_{\bm{I}}\otimes\id_{\bm{O}}$ where $\sigma\in\L(\H_{\bm{I}})$ is a quantum state. It can be shown that $\{T_i\}_i$ is a valid quantum tester if and only if there exists another quantum state $\rho\in\L(\H_{\bm{I}}\otimes\H_\text{aux})$ and a POVM $\{M_i\}_i$, $M_i\in\L(\H_\text{aux}\otimes\H_{\bm{O}})$ such that $T_i=\rho*M_i$ for all $i$. The tester formalism then provides a nice interpretation and simpler characterisation to measure-and-prepare superchannels.


\section{Transforming unitary quantum operations}
	We now present the main task analysed in this paper. Let $f:\SU(d)\to\SU(d')$ be a function which transforms unitary operators to unitary operators and $\SU(d)$ is the group of unitary $d$-dimensional operators with determinant one (special unitary group of dimension $d$). We consider a scenario where one has access to $k$ uses of an arbitrary $d$-dimensional unitary quantum operation described by an operator%
\footnote{We remark that, without loss of generality, unitary quantum operations are described by elements of $\SU(d)$, the group of unitary operators with determinant one (special unitary group of dimension $d$). This holds because for any $\phi\in\mathbb{R}$, the unitary operators $U$ and $e^{i\phi}U$ represent the same physical operation, \ie,  $U\rho U^\dagger= (e^{i\phi}U)\rho(e^{i\phi}U)^\dagger$ and $\dketbra{U}{U}=\dketbra{e^{i\phi}U}{e^{i\phi}U}$.}
 $U\in\SU(d)$. Our goal is to design a universal quantum circuit or a quantum process which approximates the transformation $U^{\otimes k}\mapsto f(U)$ for any $U\in\SU(d)$. 
	By exploiting the Choi isomorphism, the link product, and the concept of superchannels presented in the previous sections, the problem tackled in this paper can be phrased as:
	\begin{equation*}
	\boxed{
	\begin{aligned}
	&\text{Given a function } f:\SU(d)\to\SU(d'), \ \text{find the optimal (parallel/sequential/general) superchannel } \nonumber \\
	& S \text{ such that:} \qquad \qquad \qquad S*\dketbra{U}{U}^{\otimes k} \approx \dketbra{f(U)}{f(U)} \quad \forall U\in\SU(d) .
	\end{aligned}
		}
	\end{equation*}

	In the next section, we explore the concepts of robustness and channel fidelity to assign a precise meaning to $S*\dketbra{U}{U}^{\otimes k} \approx \dketbra{f(U)}{f(U)}$. Further, we establish the concept of optimality for unitary channel transformations.
	
\subsection{Quantifying the performance of unitary transformations} \label{sec:fidelity}
\subsubsection{Average fidelity}
	One way to quantify how similar two operations are is given by the channel fidelity. The fidelity between a unitary channel with Choi operator $\dketbra{U}{U} \in \L(\H_P\otimes\H_F)$ and an arbitrary channel with Choi operator $C\in \L(\H_P\otimes\H_F)$ is
\begin{align}
	F(C,\dketbra{U}{U}):=& \frac{1}{d_P^2} \dbra{U} C \dket{U}, \nonumber \\
	 =& \frac{1}{d_P^2}\tr\left( C\,\dketbra{U}{U}\right) 
\end{align}
where $d_P$ is the dimension of the linear space $\H_P$. The channel fidelity satisfies a list of operational properties that guarantees it to be a good quantifier~\cite{raginsky01}. Also, if $C$ is the Choi operator of a quantum channel $\map{C}$ and $\dketbra{U}{U}$ is the Choi operator of a unitary channel $\map{U}$, we can relate the concept of channel fidelity with the quantum state fidelity $f(\rho,\ketbra{\psi}{\psi}):=\bra{\psi}\rho\ket{\psi}$ via the identity~\cite{horodecki98}
\begin{equation} \label{eq:fidelities}
	\int_\text{Haar} f\Big(\map{C}(\ketbra{\psi}{\psi}),\, U\ketbra{\psi}{\psi}U^\dagger\Big)\dif{\ket{\psi}}
	=\frac{d}{d+1} F(C,\dketbra{U}{U}) + \frac{1}{d+1}.
\end{equation}

	Hence, a natural way to quantify the performance of a superchannel $S$ on transforming $\dketbra{U}{U}^{\otimes k}$ into $\dketbra{f(U)}{f(U)}$ is then given by its \textit{average fidelity,}
\begin{align}
		\mean{F}:=& \int_\text{Haar} F\Big(\left(S*\dketbra{U}{U}^{\otimes k}\right),\dketbra{f(U)}{f(U)}\Big) \;\dif U \nonumber \\
		=& \int_\text{Haar} \frac{1}{d^2}  \tr\Big(\left( S*\dketbra{U}{U}^{\otimes k} \right)\,\dketbra{f(U)}{f(U)}\Big) \;\dif U,
  \end{align}
where $d$ is the input dimension of $f(U)$. {{The integral is executed according to the Haar measure $\dif U$ with respect to $\SU(d)$. For the compact group $\SU(d)$, the Haar measure is 1) invariant: satisfying $\dif U=\dif (V U)=\dif (U V), \forall V\in \SU(d)$ (the unique left- and right-invariant measure on $\SU(d)$ coincide as the group is unimodular \cite{raczka1986theory}),  and 2) normalized: $\int_{\SU(d)} \dif U = 1$}}. 

This notion of optimal average fidelity is natural when we are interested in quantifying the performance on average, for any choice of $U$. It has been used for several related tasks such as quantum unitary estimation~\cite{acin00}, unitary cloning~\cite{chiribella08clone}, unitary learning~\cite{bisio10}, iteration of unitary gates~\cite{soleimanifar16}, and unitary transformations with a single use~\cite{ebler16}. The next subsections give further motivation to choose the average fidelity as a figure of merit.

\subsubsection{Worst-case fidelity}
	Given that realistic scenarios mostly take one specific $U$ as input, identifying the average performance of a superchannel might not tell much about the quality for one specific transformation $U$: the performance could be good for some choices of $U \in \SU(d)$ and bad for others.
	It is then interesting to analyse the \textit{worst-case fidelity} as lower bound of performance. Given a superchannel $S$ and a desired transformation $f(U)$, the worst-case fidelity is defined as 
\begin{align}
F_{\rm wc}:=\min_{U\in\SU(d)} F(S*\dketbra{U}{U}^{\otimes k},\dketbra{f(U)}{f(U)}).
\end{align}

	Intriguingly, for the main classes of functions $f$ considered in this work -- when $f$ is a homomorphism, \ie, $f(UV)=f(U)f(V)$ or when $f$ is an anti-homomorphism  \ie, $f(UV)=f(U)f(V)$ -- maximizing the average fidelity coincides with maximizing the worst-case fidelity. This is the content of the following.
\begin{theorem}\label{thm:WCF}
Let $S\in\L(\H_P\otimes\H_{\bm{I}}\otimes\H_{\bm{O}}\otimes\H_F)$ be a parallel/sequential/general superchannel that transforms $k$ uses of a unitary operator $U\in\SU(d)$ into $f(U)\in\SU(d')$ with average fidelity $\mean{F}$. If $f$ is a homomorphism, \ie, $f(UV)=f(U)f(V)$, or an anti-homomorphism, \ie, $f(UV)=f(V)f(U)$, there exists a parallel/sequential/general superchannel $S'\in\L(\H_P\otimes\H_{\bm{I}}\otimes\H_{\bm{O}}\otimes\H_F)$ that transforms $k$ uses of a unitary operator $U$ into $f(U)$ with worst-case fidelity $F_{\rm wc}=\mean{F}$.
\end{theorem} 
	The proof of this proposition is presented in App.~\ref{app:f_is_rep} and combines the covariant properties of superchannels discussed in Sec.~\ref{sec:f_isREP} Sec.~\ref{sec:f_T_isREP} with Holevo's argument for covariant averaging~\cite{holevoBook}.

\subsubsection{Optimal white noise visibility}
	Another pertinent figure of merit is the \textit{white noise visibility}, a quantifier which appears naturally in cloning quantum states \cite{werner98} and robustness of entanglement \cite{vidal98}. 	A superchannel $S$ is said to have a white noise visibility $\eta$ for transforming $k$ uses of $U$ into $f(U)$ if we can write
\begin{align} \label{eq:white_noise}
S*\dketbra{U}{U}^{\otimes k} = \eta \ketbra{f(U)}{f(U)} + (1-\eta) \id_\text{P}\otimes\frac{\id_\text{F}}{d}, \quad \quad \forall U \in\SU(d) .
\end{align}

	If $S$ is a superchannel that transforms $k$ uses of $U$ into $f(U)$ with white noise visibility $\eta$, its average fidelity $\mean{F}$ is given by
\begin{align} \label{eq:white_is_fidelity}
	\mean{F} &=\frac{1}{d^2} \int_\text{Haar}   \dbra{f(U)}\left( S*\dketbra{U}{U}^{\otimes k} \right) \dket{f(U)} \;\dif U \nonumber \\
		&= \frac{1}{d^2} \int  \dbra{f(U)}\left( \eta \dketbra{f(U)}{f(U)} + (1-\eta) \id_\text{P}\otimes\frac{\id_\text{F}}{d} \right) \dket{f(U)} \;\dif U \nonumber \\
			&= \frac{1}{d^2} \int   \eta
\left( \dbra{f(U)}\dketbra{f(U)}{f(U)}\dket{f(U)}\right) + 
(1-\eta) \left(\frac{1}{d}\dbra{f(U)}\dket{f(U)} \right)\;\dif U \nonumber  \\	
		&= \frac{1}{d^2} \int   \left(\eta d^2 + (1-\eta)  \right) \;\dif U \nonumber \\
		& = \eta + \frac{1-\eta}{d^2} .
\end{align}
	Note however, due to the restrictive structure of the white noise visibility in Eq.~\eqref{eq:white_noise}, in general, it is not possible to establish a one-to-one connection with the average fidelity. In order to have a well-defined white noise visibility $\eta$, the superchannel $S$ has to operate on $U$ in a very particular form. However, in a large class of functions $f$ considered in this work, the white noise fidelity coincides with the average fidelity -- meaning that the structure given in Eq.~\eqref{eq:white_noise} can be assumed for the action of $S$ without loss of generality.

\begin{theorem} \label{theo:white_noise_cov} 
	 Let $f:\mathcal{SU}(d)\to\mathcal{SU}(d)$ be a function respecting
\begin{itemize}
\item $f(UV)=f(U)f(V)$ for all $U,V\in \mathcal{SU}(d)$, or $f(UV)=f(V)f(U)$  $U,V\in \mathcal{SU}(d)$
\item $f(U^*)=f(U)^*$
\item For the Haar measure, the differential $\dif U$ is invariant under the substitution $\dif U\to \dif f(U)$
\end{itemize}	 
If $S$ is a parallel/sequential/general superchannel transforming $k$ uses of $U$ into $f(U)$ with average fidelity $\mean{F}$,
there exists a parallel/sequential/general superchannel $S'$ such that
\begin{equation}
	S'* \dketbra{U}{U}^{\otimes k}_{\bm{IO}} = \eta \dketbra{f(U)}{f(U)}_{PF} + (1-\eta) \id_P\otimes \frac{\id_F}{d},
\end{equation}	 
where $\mean{F}=\eta + \frac{1-\eta}{d^2}$.
\end{theorem}
	The proof of Thm.~\ref{theo:white_noise_cov} is presented in Appendix~\ref{app:wnc}.

\subsubsection{Probabilistic exact strategies can be converted to deterministic non-exact ones} \label{sec:prop}
	While in this work we propose to analyse on deterministic non-exact transformations, it is possible as well to consider probabilistic exact transformations. In the probabilistic exact setting, it is possible to reject the outcome, meaning that a certain probability of failure is accepted \cite{quintino19PRA,quintino19PRL,dong20}. Here, we consider the case where the output is rejected whenever the protocol fails to execute the desired transformation exactly (with unit fidelity). 
	
	To study probabilistic exact strategies more formally, it is necessary to introduce the concept of \textit{quantum superinstruments}. Superinstruments may be seen as a higher-order version of quantum instruments to networks. Concretely, a superinstrument is a set of positive semidefinite maps $\{S_i\}_i$ such that $S:=\sum_iS_i$ is a superchannel. In the binary case of ``succeed" or ``fail", we say that a superinstrument $\{S_s,S_f\}$ transforms $k$ uses of a unitary $U$ into $f(U)$ with probability of success probability $p_s$ if
\begin{align} \label{eq:probabilistic_exact}
	 S_s*\dketbra{U}{U}^{\otimes k} = p_s \ketbra{f(U)}{f(U)} \quad \quad \forall U \in\SU(d) .
\end{align}

	Probabilistic exact strategies that transform $k$ uses of $U$ into $f(U)$ with success probability $p_s$ provide us a simple method to construct deterministic exact superchannels with average fidelity $\mean{F}\geq p_s$. If $\{S_s,S_f\}$ is a parallel/sequential/general superinstrument, by definition $S:=S_s+S_f$ is a parallel/sequential/general superchannel which respects
\begin{align}
	\mean{F} &=\frac{1}{d^2} \int   \dbra{f(U)}\left( S*\dketbra{U}{U}^{\otimes k} \right) \dket{f(U)} \dif U \nonumber \\
	&=\frac{1}{d^2} \int  \Big( \dbra{f(U)} S_s*\dketbra{U}{U}^{\otimes k} \dket{f(U)} 
	+ \dbra{f(U)}  S_f*\dketbra{U}{U}^{\otimes k}   \dket{f(U)}\Big) \dif U \nonumber \\
	&=\frac{1}{d^2} \int  \Big( p_s  \dbra{f(U)} \dketbra{f(U)}{f(U)} \dket{f(U)} 
	+ \dbra{f(U)}  S_f*\dketbra{U}{U}^{\otimes k}   \dket{f(U)}\Big) \dif U \nonumber \\
&=\frac{1}{d^2} \int  \Big( p_s d^2
	+ \dbra{f(U)}  S_f*\dketbra{U}{U}^{\otimes k}   \dket{f(U)}\Big) \dif U \nonumber \\
	&=p_s + \frac{1}{d^2} \int  \Big( \dbra{f(U)}  S_f*\dketbra{U}{U}^{\otimes k}   \dket{f(U)}\Big) \dif U \nonumber \\			
		& \geq p_s . 
\end{align}

	Note however that deterministic non-exact protocols may not lead to probabilistic exact ones. Unitary transposition is one concrete example discussed later in Sec.~\ref{sec:f_T_isREP}: we show that optimal parallel strategies can always be implemented by a measure-and-prepare strategy. Yet,  in a probabilistic exact case, non-coherent strategies based on measure-and-prepare superchannels will necessarily lead to $p_s=0$.  Additionally, when a single use is considered ($k=1$), approximate cloning of unitary operations \cite{chiribella08clone} and unitary complex conjugation  for $d>2$ \cite{ebler16} are possible with non-zero average fidelity, but they cannot be done in a probabilistic exact way \cite{miyazaki17,quintino19PRA} -- meaning $p_s=0$.
	
\subsection{Performance operator and semidefinite programming} \label{sec:SDP}
	As first observed in Ref.~\cite{ebler16}, when seeking for the optimal superchannels to maximise the average fidelity for a desired transformation, $\dketbra{U}{U}^{\otimes k} \mapsto \dketbra{f(U)}{f(U)} $ it is convenient to define the \textit{performance operator}:
\begin{align}
		\Omega:=&\frac{1}{d^2} \int_\text{Haar} \dketbra{f(U)}{f(U)}_{PF}  \otimes \dketbra{U^*}{U^*}^{\otimes k}_{\bm{IO}} \;\dif U 
\end{align}
	where $U^*$ is the complex conjugate of the operator $U$ written in the computational basis. The performance operator is useful to evaluate the average fidelity performance of a superchannel $S$ via the relation $\mean{F}=\tr(S\,\Omega)$ which holds true since
\begin{align}
		\tr (S \Omega) & =  \frac{1}{d^2} \int_\text{Haar}\tr\Big[ S\,\left( \dketbra{f(U)}{f(U)}_{PF}  \otimes \dketbra{U^*}{U^*}^{\otimes k}_{\bm{IO}}\right) \Big]\dif U  \nonumber \\
		&=  \frac{1}{d^2} \int_\text{Haar} \tr\Big[\left(S*\dketbra{U}{U}^{\otimes k}_{\bm{IO}}\right)\, \dketbra{f(U)}{f(U)}_{PF}\, \Big] \dif U \label{eq:UstarTrick} \nonumber \\
		&=\mean{F},
\end{align}
where in Eq.~\eqref{eq:UstarTrick} we have used the identity $\dketbra{U^*}{U^*}=\dketbra{U}{U}^T.$

	For any given performance operator $\Omega$, the problem of maximising the fidelity over a set $\set S$ of superchannels with $k$-slots can be phrased as 
\begin{align} \label{eq:SDP}
	&\max_{S {{\in \set S_\alpha}}} \tr(S\Omega) \ .
\end{align}
	The subscript $\alpha$ of the set $\set S_\alpha$ labels the desired set, \textit{i.e.,} $\alpha\in \{$parallel, sequential, general, prepare-and-measure$\}$ In Sec.~\ref{sec:k-slots} we see that -- apart from the prepare-and-measure case -- superchannels can be completely described by linear and positive semidefinite constraints. Hence, for such cases, Eq.\eqref{eq:SDP} is a semidefinite programming (SDP) problem.
	Reference \cite{ebler16} also showed that the dual problem of the SDP presented in Eq.~\ref{eq:SDP} is given by,
\begin{align}
	\min_{\overline{S}\in\overline{\set{S_\alpha}}} \, & \lambda \nonumber \\
	\text{such that: } \, & \Omega \leq \lambda \overline{S} \ .
\end{align}
Here, $\overline{\set{S_\alpha}}$ stands for the dual affine of the set of the desired $k$-slot superchannels $\set S_\alpha$. When $\set{S}$ is a set of linear operators, its dual affine set $\overline{\set{S}}$ is the set containing all operators $\overline{S}$ such that $\tr(\overline{S}S)=1\, \forall \, S\in \set{S}$. For a more detailed discussion on dual affine sets of superchannels and examples, we refer the reader to Ref.~\cite{ebler16} and Ref.~\cite{bavaresco21}.
	

\section{Revisiting the homomorphic case: $f(UV)=f(U)f(V)$ } \label{sec:f_isREP}
	In this section we recap known results for the case where $f:\SU(d)\to\SU(d')$ is a homomorphism, that is, it respects $f(UV)=f(U)f(V)$. This corresponds to the case where the function $f$ is a unitary representation of $\SU(d)$, class of transformations analysed in Ref.~\cite{bisio10}. {{ Here we remark that, when the dimension $d$ is equal to $d'$, up to a unitary equivalence, the only non-trivial homomorphism $f:\SU(d)\to\SU(d)$ is the unitary complex conjugation $f(U)=U^*$. This follow from the fact that there exists only three $d-$dimensional representations for the group $\SU(d)$ \cite{harrisBook}, the trivial representation $f_\text{trivial}(U)=\id$, the defining representation, $f_\text{def}(U)=U$, and the conjugate representation $f_\text{conj}(U)=U^*$.}} 

	When $f(UV)=f(U)f(V)$, the invariance of Haar measure ensures that $\Omega$ respects the invariant relation
\small
	\begin{align} \label{eq:omega_invariance}
&\Big(\id_P \otimes \id_{\bm{I}} \otimes V_{\bm{O}}^{*^{\otimes k}} \otimes f(V)_F\Big)
\Omega
\Big(\id_P \otimes \id_{\bm{I}} \otimes V_{\bm{O}}^{*^{\otimes k}} \otimes f(V)_F\Big)^\dagger \nonumber  \\
&=\frac{1}{d^2} \int \Big(\id_P \otimes \id_{\bm{I}} \otimes V_{\bm{O}}^{*^{\otimes k}} \otimes f(V)_F\Big)
\dketbra{f(U)}{f(U)}_{PF} \otimes \dketbra{U^*}{U^*}^{\otimes k}_{\bm{IO}}
\Big(\id_P \otimes \id_{\bm{I}} \otimes V_{\bm{O}}^{*^{\otimes k}} \otimes f(V)_F\Big)^\dagger \text{d}U \nonumber  \\
&=\frac{1}{d^2} \int 
\dketbra{f(V)f(U)}{f(V)f(U)}_{PF} \otimes \dketbra{V^*U^*}{V^*U^*}^{\otimes k}_{\bm{IO}} \; \text{d}U \nonumber \\
&=\frac{1}{d^2} \int 
\dketbra{f(VU)}{f(VU)}_{PF} \otimes \dketbra{(VU)^*}{(VU)^*}^{\otimes k}_{\bm{IO}}\; \text{d}U \nonumber \\
&=\Omega.
\end{align}
\normalsize
Hence, the performance operator respects the commutation relations
\begin{align} \label{eq:f(u)_commutation}
	[\Omega, \id_P \otimes \id_{\bm{I}} \otimes U_{\bm{O}}^{*^{\otimes k}} \otimes f(U)_F]&=0 \quad \forall U \in\SU(d),  \\
	[\Omega, f(U)_P \otimes U_{\bm{I}}^{*^{\otimes k}} \otimes \id_{\bm{O}} \otimes \id_F]&=0 \quad \forall U \in\SU(d), \label{eq:f(u)_commutation2}
\end{align}
where the second relation can be obtained with the aid of the identity $\id\otimes U\dket{\id}=U^T\otimes \id\dket{\id}$ and the fact that the relation $\forall U^T\in\SU(d)$ is equivalent to $\forall U\in\SU(d)$.
	In Appendix\,\ref{app:f_is_rep} we show to exploit these commutation relations to write
\begin{equation}
	\Omega = \frac{1}{d^2}\sum_i \frac{\left(P^i_{\bm{I}P}\right)^* \otimes P^i_{\bm{O}F}}{d_i}.
\end{equation}
where $\big\{P^i\big\}_i$ is an orthogonal basis\footnote{A set of operators $\big\{P^i\big\}_i$ is orthogonal if $\tr(P{^i}^\dagger P^j)=0$ when $i\neq j$.} for the linear space spanned by operators 
$P\in\L(\mathbb{C}_d^{\otimes k}\otimes\mathbb{C}_d)$   respecting 
$[P, {U^*}^{\otimes k}\otimes f(U)]=0$ for all $U\in\SU(d)$ and $d_i:=\tr(P^i {P^i}^\dagger)$. One way to obtain such basis is to find the isometries between the equivalent irreducible representations of the group given by $U^{*^{\otimes k}}\otimes f(U)$ \cite{harrisBook}. In Appendix~\ref{app:explicit} we present examples and explicit constructions for the case $f(U)=U^*$.

	In Ref.~\cite{bisio13} the authors show that when $f$ is a homomorphism, every superchannel admits a parallel implementation without decreasing the average fidelity. In other words, sequential strategies and even general indefinite causal order strategies cannot outperform parallel ones. We state this main result of Ref.~\cite{bisio13} here and present an alternative proof in Appendix~\ref{app:par_is_optimal}.
\begin{figure}[h!] \label{fig:par_is_optimal_for_rep}
\begin{center}
	\includegraphics[scale=0.42]{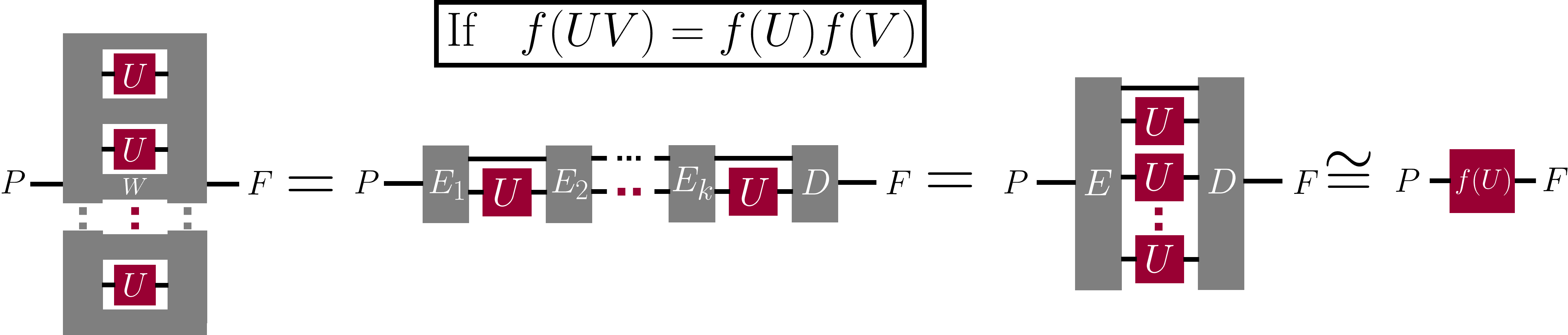} 
\end{center}
\caption{If the function $f$ respects $f(UV)=f(V)f(U)$, every protocol that transform $k$ copies of $U$ into $f(U)$ can be made by a parallel circuit with the same average fidelity \cite{bisio13}. Examples of functions respecting this property are unitary complex conjugation $f(U)=U^*$ and unitary cloning $f(U)=U\otimes U$.}
\end{figure}

\begin{proposition} [Ref.~\cite{bisio13}] \label{prop:f_is_rep}
Let $S\in\L(\H_P\otimes\H_{\bm{I}}\otimes\H_{\bm{O}}\otimes\H_F)$ be a general superchannel that transforms $k$ uses of a unitary operator $U$ into $f(U)$ with average fidelity $\mean{F}$. If $f(UV)=f(U)f(V)$, there exists a parallel superchannel $S'\in\L(\H_P\otimes\H_{\bm{I}}\otimes\H_{\bm{O}}\otimes\H_F)$ that transforms $k$ uses of a unitary operator $U$ into $f(U)$ with average fidelity $\mean{F}$ and respects the commutation relation
\begin{equation}
	[S', f(A)_P \otimes A_{\bm{I}}^{*^{\otimes k}} \otimes B_{\bm{O}}^{*^{\otimes k}} \otimes f(B)_F]=0,
\end{equation}
for any unitary operators $A,B$.

	Additionally, the action of $S'$ on $\dketbra{U}{U}^{\otimes k}$ is described by its action on the identity channel $\dketbra{\id}{\id}^{\otimes k}$ via
\begin{equation} \label{eq:S_is_covariant}
S'*\dketbra{U}{U}_{\bm{IO}}^{\otimes k} = \Big(\id_P \otimes f(U)_F\Big)
\Big( S'*\dketbra{\id}{\id}^{\otimes k}_{\bm{IO}}\Big)
\Big(\id_P \otimes f(U)_F\Big)^\dagger .
\end{equation}
\end{proposition}

\subsection{Unitary complex conjugation}

	This subsection addresses the case where $f$ stands for the unitary complex conjugation in the computational basis, that is, $f(U)=U^*$. We first point that $(UV)^*=U^*V^*$, hence, unitary complex conjugation function is a homomorphism. 
	
	The performance operator for unitary complex conjugation is given by
\begin{align}
		\Omega=&\frac{1}{d^2} \int_\text{Haar} \dketbra{U^*}{U^*}_{PF}  \otimes \dketbra{U^*}{U^*}^{\otimes k}_{\bm{IO}} \;\dif U  \nonumber  \\
		=&\frac{1}{d^2} \int_\text{Haar} \dketbra{U}{U}_{PF}  \otimes \dketbra{U}{U}^{\otimes k}_{\bm{IO}} \;\dif U,\label{eq:performance_conj}  
\end{align}
and respects the commutation relations
\begin{align}
	[\Omega, A_P \otimes A_{\bm{I}}^{^{\otimes k}} \otimes B_{\bm{O}}^{^{\otimes k}} \otimes B_F]&=0 \quad \forall A,B \in\SU(d)
\end{align}
	Following Eq.~\ref{eq:performance_conj}, the performance operator can be evaluated by finding a basis for operators commuting with $U^{\otimes(k+1)}$. For the particular case of $k=1$, an explicit form of the performance operator can be found at Ref.~\cite{ebler16}. In Appendix~\ref{app:explicit} we present an explicit form for the case $k=1$ and $k=2$.

	When only a single used is allowed ($k=1$), Ref.~\cite{ebler16} shows that the optimal average fidelity is given by
		\begin{equation}
		\mean{F} =\frac{2}{d(d-1)} \ . \label{eq:cjk1}
	\end{equation}
Also, for the case where $k=d-1$, Ref.~\cite{miyazaki17} presents an explicit parallel superchannel which performs the  transformation $U^{\otimes (d-1)}\mapsto U^*$ exactly, meaning that $\mean{F}=1$.

	We finish this subsection by summarising the main results about optimal unitary complex conjugation:
\begin{itemize}
\item (Ref.~\cite{ebler16}) For $k=1$ use, the optimal average fidelity for $U\mapsto U^*$ is $\mean{F} =\frac{2}{d(d-1)}$
\item (Ref.~\cite{miyazaki17}) When $k=d-1$, the transformation  $U^{\otimes k}\mapsto U^*$ can be obtained exactly
\item (Ref.~\cite{bisio10}) For any number of uses $k$, the optimal average fidelity is attainable by a parallel superchannel
\end{itemize}


\section{Analysing the anti-homomorphic case: $f(UV)=f(V)f(U)$  } \label{sec:f_T_isREP}
		We now consider the class of functions $f$ with ``inverts'' the operator composition. A function $f:\SU(d)\to\SU(d')$ is an anti-homomorphism, if $f$ respects
\begin{equation} \label{eq:not_group}
 f(UV)=f(V)f(U), \quad \forall U,V\in\SU(d),
\end{equation}
or, equivalently,
\begin{equation}
	f(UV)^T=f(U)^T f(V)^T, \quad \forall U,V\in\SU(d).
\end{equation}
{{We remark that, when the dimension $d$ is equal to $d'$, up to a unitary equivalence, the only non-trivial anti-homomorphisms $f:\SU(d)\to\SU(d)$ are the unitary transposition $f_\text{trans}(U)=U^T$ and the unitary inversion $f_\text{inv}(U)=U^{\dagger}$ (remark is also noted by Ref.~\cite{chiribella20timeflip}). This follow from the fact that there exists only three $d-$dimensional representations for the group $\SU(d)$ \cite{harrisBook}, the trivial representation $f_\text{trivial}(U)=\id$, the defining representation, $f_\text{def}(U)=U$, and the conjugate representation $f_\text{conj}(U)=U^*$.}}

	When $f(UV)=f(V)f(U)$, calculations analogous to Eq.~\eqref{eq:omega_invariance} shows that the performance operator respects the commutation relations
\begin{align}
	[\Omega, f(U)^T_P \otimes \id_{\bm{I}} \otimes U_{\bm{O}}^{*^{\otimes k}} \otimes \id_F]&=0 \quad\quad \forall U \in\SU(d)
	\label{eq:f(u)T_commutation}\\
	[\Omega, \id_P \otimes U_{\bm{I}}^{*^{\otimes k}} \otimes \id_{\bm{O}} \otimes f(U)_F^T ] &=0  \quad\quad \forall U \in\SU(d) \ .
	\label{eq:f(u)T_commutation2}
\end{align}
By exploiting these relations, in Appendix~\ref{app:fT_is_rep} we show that this performance operator can be expressed as
\begin{equation}\label{eq:perfomance_anti_homo}
\Omega = \frac{1}{d^2} 
\sum_i \frac{ \left(P^i_{\bm{O}P}\right)^* \otimes P^i_{\bm{I}F}}   {d_i} \ ,
\end{equation}
where $\big\{P^i\big\}_i$ is an orthogonal basis for the linear space spanned by linear operators $P\in\L(\mathbb{C}_d^{\otimes k}\otimes \mathbb{C}_d)$ commuting with $U^{*^{\otimes k}}\otimes f(U)^T$ and $d_i:=\tr(P^i {P^i}^\dagger)$. One way to obtain this basis is by finding the isometries between equivalent irreducible representations of the group given by $U^{*^{\otimes k}}\otimes f(U)^T$. In Appendix~\ref{app:explicit} we present examples and explicit constructions for the case $f(U)=U^T$ and $f(U)=U^{-1}$.

	Similarly to the homomorphic case, superchannels used to implement anti-homomorphic unitary transformations enjoy useful covariant properties.
\begin{lemma} \label{lemma:fT_is_rep}
	Let $S\in\L(\H_P\otimes\H_{\bm{I}}\otimes\H_{\bm{O}}\otimes\H_F)$ be a parallel/sequential/general superchannel that transforms $k$ uses of a unitary operator $U$ into $f(U)$ with average fidelity $\mean{F}$. If $f(UV)=f(V)f(U)$, there exists a parallel/sequential/general superchannel $S'\in\L(\H_P\otimes\H_{\bm{I}}\otimes\H_{\bm{O}}\otimes\H_F)$ that transforms $k$ uses of a unitary operator $U$ into $f(U)$ with average fidelity $\mean{F}$ and respects the commutation relation
\begin{equation}
	\Big[S', f(A)_P^T \otimes B_{\bm{I}}^{*^{\otimes k}} \otimes A_{\bm{O}}^{*^{\otimes k}} \otimes f(B)^T_F\Big]=0,
\end{equation}
for any unitary operators $A,B$.

	Additionally, the action of $S'$ on $\dketbra{U}{U}^{\otimes k}$ is described by its action on the identity channel $\dketbra{\id}{\id}^{\otimes k}$ via
\begin{equation} \label{eq:S_is_covariant2}
	S'*\dketbra{U}{U}^{\otimes k} = \Big(\id_P \otimes {f(U^T)}^T_F\Big)
	 \Big( S'*\dketbra{\id}{\id}^{\otimes k}_{\bm{IO}}\Big)
	 \Big(\id_P \otimes {f(U^T)}^T_F\Big)^\dagger .
\end{equation}
\end{lemma}
	The proof of this Lemma can be found in Appendix~\ref{app:fT_is_rep}.

\subsection{Unitary transposition}
	This subsection addresses the case where $f$ stands for the unitary transposition in the computational basis, that is, $f(U)=U^T$. We first point that $(UV)^T=V^TU^T$, hence, unitary transposition function is an anti-homomorphism.
	
	The performance operator for unitary transposition is given by
\begin{align}
		\Omega=&\frac{1}{d^2} \int_\text{Haar} \dketbra{U^T}{U^T}_{PF}  \otimes \dketbra{U^*}{U^*}^{\otimes k}_{\bm{IO}} \;\dif U  \nonumber \\
		=&\frac{1}{d^2} \int_\text{Haar} \dketbra{U}{U}_{PF}  \otimes \dketbra{U^\dagger}{U^\dagger}^{\otimes k}_{\bm{IO}} \;\dif U, 
\end{align}
and respects the commutation relations
\begin{align}
	[\Omega, A_P \otimes B_{\bm{I}}^{*^{\otimes k}} \otimes A_{\bm{O}}^{*^{\otimes k}}  \otimes B_F]&=0 \quad \forall A,B \in\SU(d).
\end{align}
	With the aid of Eq.~\eqref{eq:perfomance_anti_homo}, the performance operator can be evaluated by finding a basis for operators commuting with $U^{*^{\otimes k}}\otimes U$. In Appendix~\ref{app:explicit} we present an explicit form for the case where $k=1$ and $k=2$.

	\subsubsection{Parallel unitary transposition and estimating unitary operations} \label{sec:trans=est}
	Let us consider the problem in which a unitary operation is uniformly sampled from $\SU(d)$ and our goal is to guess which unitary operation was sampled by performing at most $k$ used on it. An estimation strategy for this problem consists in preparing a bipartite state $\ket{\phi}\in\H_{\bm{I}}\otimes\H_{\text{aux}}$, which will be subject it to $k$ parallel%
	\footnote{When estimating unitaries which are uniformly sampled from a group, in our case, $\SU(d)$, adaptive and general strategies cannot provide an advantage over parallel ones~\cite{chiribella08memory_effects,bavaresco21b}. Hence, we can assume that the $k$ uses of the input operation are performed in parallel with no loss of generality.} 
	uses of an arbitrary unitary operator $U:\H_{\bm{I}}\to\H_{\bm{O}}$ to become 
	\begin{equation}
		\ket{\phi_U}:=\big(U^{\otimes k}_{\bm{I}}\otimes \id_\text{aux}\big)\ket{\phi}\in\H_{\bm{O}}\otimes\H_\text{aux}.
	\end{equation} 
	Then, we perform quantum measurement with POVM element $M_i\in\L(\H_{\bm{O}}\otimes\H_\text{aux})$ on $\ket{\phi_U}$ and, accordingly to the outcome $i$ of this measurement, we guess that the input unitary is described by $U_i$.
In terms of average fidelity, the performance of estimating uniform unitaries over $\SU(d)$ is given by: 
	\begin{align}
		\mean{F}_\text{ext}:=\int_{U\in\SU(d)} \sum_i 
		\tr\Big(M_i \ketbra{\phi_U}{\phi_U}\Big)\,
		 F\Big(\dketbra{U_i}{U_i},\dketbra{U}{U}\Big) \text{d}U .
	\end{align}
	
	Unitary estimation strategies are intimately related to measure-and-prepare superchannels defined in Sec.~\ref{sec:measure-and-prep_and_delayed}. More precisely, by defining 
	$S:=\sum_i\ketbra{\phi}{\phi}_{\bm{I}\text{aux}}*(M_i^T)_{\bm{O}\text{aux}}*\dketbra{U_i}{U_i}_{PF}$,
	and using the relation $\tr(A_1)\tr(B_2)=\tr(A_1\otimes B_2)=\tr(A_1*B_2)$
	we obtain 
\begin{align}
\mean{F}_\text{ext}:=&\int_{U\in\SU(d)} \sum_i 
		\tr\Big(M_i \ketbra{\phi_U}{\phi_U}\Big)\,
		 F\Big(\dketbra{U_i}{U_i},\dketbra{U}{U}\Big) \text{d}U  \nonumber \\
=& \frac{1}{d^2}\int
\sum_i \tr\Big(\left( M_{i} \right)_{\bm{O}\text{aux}} \, \dketbra{U}{U}_{\bm{IO}}^{\otimes k}* \ketbra{\phi}{\phi}_{\bm{I}\text{aux}}\Big)
		\tr\Big(\dketbra{U_i}{U_i}_{PF}\,\dketbra{U}{U}_{PF}\Big) \text{d}U  \nonumber \\
=& \frac{1}{d^2}\int
		\tr\Big(  \sum_i 
		\left( M_{i}^T \right)_{\bm{O}\text{aux}} * \dketbra{U}{U}_{\bm{IO}}^{\otimes k}* \ketbra{\phi}{\phi}_{\bm{I}\text{aux}}*
		 \dketbra{U_i}{U_i}_{PF}\,	\dketbra{U}{U}_{PF}\Big) \text{d}U \nonumber \\	
=& \frac{1}{d^2}\int 
		\tr\Big( S * \dketbra{U}{U}_{\bm{IO}}^{\otimes k}\; \dketbra{U}{U}_{PF}\Big) \text{d}U \nonumber \\			
=& \int
		 F\Big(S*\dketbra{U}{U}_{\bm{IO}}^{\otimes k},\dketbra{U}{U}_{PF}\Big) \text{d}U.
\end{align} 
\normalsize
	Hence, the problem of estimating unitary operations may be viewed as transforming a unitary operation to itself with a measure-and-prepare strategy. In the language of this paper, this corresponds to implementing  the identity function, $f(U)=U$, with parallel measure-and-prepare superchannels.
	
We now show that when considering parallel strategies, unitary transposition is equivalent to the problem of estimating unitary operations uniformly sampled in $\SU(d)$.

\begin{figure}[h!] \label{fig:par=estimation}
	\begin{center}
		\includegraphics[scale=0.5]{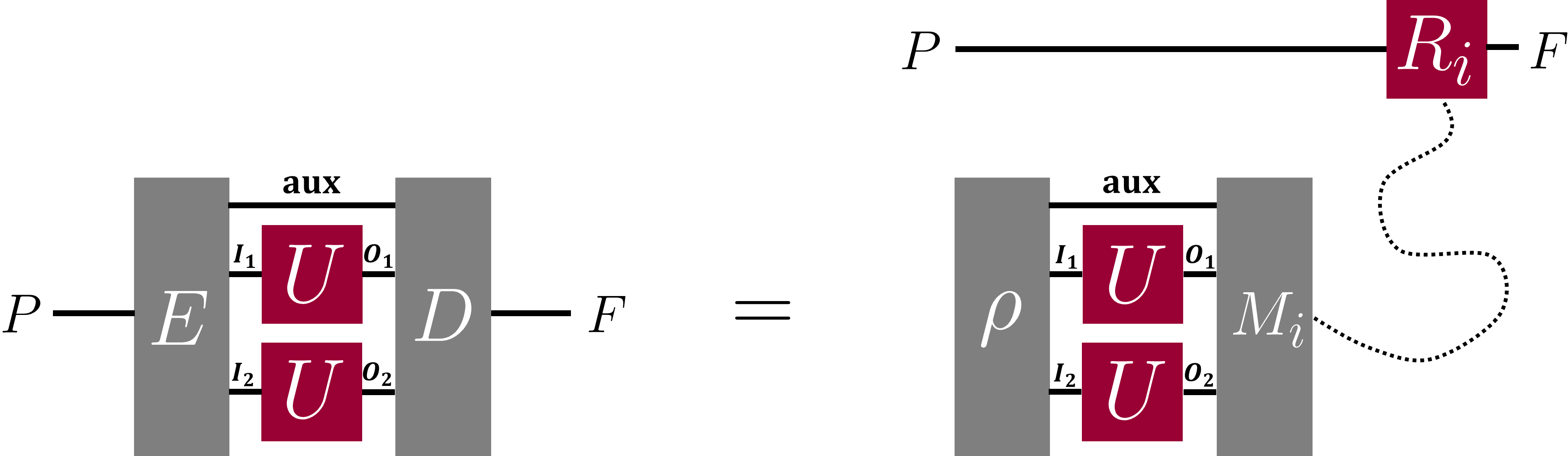} 
	\end{center}
\caption{Every parallel strategy for unitary transposition and unitary inversion can be implemented by a prepare-and-measure strategy without changing its average fidelity performance. Additionally, the performance of parallel unitary transposition and parallel unitary inversion are both equivalent to the performance of estimating unitary quantum operations uniformly sampled in $\SU(d)$.}
\end{figure}

\begin{theorem} \label{theo:PAR=ESTtrans}
	The optimal average fidelity for parallel unitary transposition can always be attained by measure-and-prepare strategies. Moreover, the value of the optimal average fidelity for parallel unitary transposition is exactly the value of the optimal average fidelity for estimating unitary operations uniformly sampled in $\SU(d)$. 
\end{theorem}

\begin{proof}
	Let $S\in\L(\H_P\otimes\H_{\bm{I}}\otimes\H_{\bm{O}}\otimes\H_F)$ be a parallel superchannel which implements the desired transformation with average fidelity $\mean{F}$. By definition of parallel channels, $S$ must respect
\begin{equation}
		\tr_F(S) = \tr_{\bm{O}F}(S)\otimes \frac{\id_{\bm{O}}}{d_{\bm{O}}}.
\end{equation}
	Since $f(U)=U^T$ is an anti-homomorphic function, Lemma~\ref{lemma:fT_is_rep} ensures that, without loss of generality, $S$ respects the commutation relations
\begin{equation} \label{eq:antiHOMO}
	\Big[S, f(A)_P^T \otimes B_{\bm{I}}^{*^{\otimes k}} \otimes A_{\bm{O}}^{*^{\otimes k}} \otimes f(B)^T_F\Big]=0.
		\quad\quad \forall A,B\in\SU(d),
\end{equation}
	which implies
\begin{equation}
	 [f(A)_P^T \otimes B_{\bm{I}}^{*^{\otimes k}},\tr_{\bm{O}F}(S) ]=0,
	 	\quad\quad \forall A,B\in\SU(d).
\end{equation}
	And since $f(A)^T=A$, Schur's lemma ensures that $ \tr_{\bm{O}F}(S)= \frac{\id_P}{d_P}\otimes\tr_{P\bm{O}F}(S)$.
		
	Let $\H_{\bm{I}'}$ be an auxiliary space which is isomorphic to $\H_{\bm{I}}$. We define the quantum state $\rho:=\tr_{POF}(S)/\tr(S)$, and $\ket{\phi}\in \H_{\bm{I}}\otimes\H_{\bm{I}'}$ as be the purification of $\rho$, that is
\begin{equation}
	\ket{\phi}:= \Big(\sqrt{\rho}_{\bm{I}}\otimes \id_{\bm{I}'}\Big) \dket{\id}_{\bm{II}'},
\end{equation}
	where $\sqrt{\rho}$ is the unique positive semidefinite square root of $\rho$.
	Also, we define a quantum channel $\map{R}:\L(\H_P\otimes \H_{\bm{I}'}\otimes\H_{\bm{O}})\to\H_F$ by its Choi operator,
\begin{equation}
	R:= \Big(\id_P\otimes\sqrt{\rho}^{-1}_{\bm{I}'} \otimes \id_{\bm{O}F}\Big)\, S_{P\bm{I}'\bm{O}F} \, \Big(\id_P\otimes\sqrt{\rho}^{-1}_{\bm{I}'} \otimes \id_{\bm{O}F}\Big)^\dagger,
\end{equation}
where $\sqrt{\rho}^{-1}$ is the Moore-Penrose pseudoinverse of $\sqrt{\rho}$.
	Direct calculation shows that $S=\dketbra{\phi}{\phi}*R$, hence, instead of having an encoder channel and decoder channel construction $S=E*D$, can understand the parallel superchannel $S$ as preparing a quantum state $\ket{\phi}$, sending it over the input operation, then recovering the output state with the channel $\map{R}$. This class of parallel superchannels is known as learning process~\cite{bisio10} or delayed input-state superchannels~\cite{quintino19PRA}.
	
	Since $S$ is a learning process which respects the commutation relations presented in Eq.~\eqref{eq:antiHOMO}, we can follow the same steps of Thm.~1 of Ref.~\cite{bisio10} that $S$ can be implemented by a measure-and-prepare strategy. Additionally, since the quantum state $\rho:=\tr_{POF}(S)/\tr(S)$ respects the commutation relation
$
		[\rho, U^{\otimes k}] = 0,
$
for every  $U\in\SU(d)$, we also have\footnote{
A set of diagonalisable operators commute if and only if they are diagonal in the same basis. Hence, let $\{\ket{i}\}_i$ be an orthonormal basis in which $\rho$ and $U^{\otimes k}$ are diagonal. We can then write $\rho=\sum_i p_i \ketbra{i}{i}$ and verify that $\sqrt{\rho}=\sum_i \sqrt{p_i}\ketbra{i}{i}$ is also diagonal in the basis $\{\ket{i}\}_i$. }
 $[\sqrt{\rho}, U^{\otimes k}] = 0$ for every  $U\in\SU(d)$. Hence, the purified state $\ket{\phi}$ respects the identity
\begin{align}
	U^{\otimes k} \otimes \id \ket{\phi} 
	&=U^{\otimes k} \otimes \id \Big(\sqrt{\rho}\otimes \id\Big) \dket{\id}  \nonumber \\
	&= \Big(\sqrt{\rho}\otimes \id\Big)U^{\otimes k} \otimes \id  \dket{\id}  \nonumber \\
	&= \Big(\sqrt{\rho}\otimes \id\Big) \id \otimes {U^T}^{\otimes k}  \dket{\id}  \nonumber \\
	&=  \id \otimes {U^T}^{\otimes k} \Big(\sqrt{\rho}\otimes \id\Big)  \dket{\id}  \nonumber \\
	&=  \id \otimes {U^T}^{\otimes k} \ket{\phi}
\end{align}
for every $U\in\SU(d)$. Hence, by swapping the spaces $\H_{\bm{O}}$ and $\H_\text{aux}$ in the measurement $\{M_i\}_i$ of the measure-and-prepare strategy for unitary estimation, the average fidelity of the protocol is given by 
	$\mean{F}= \int F\Big(\dketbra{U}{U},S*\dketbra{U}{U}\Big) \text{d}U$, which is identical to the fidelity of unitary estimation protocols. We have then an method to convert parallel unitary transposition into an estimation strategy with the same average fidelity.
	
	Now, we point out that in optimal for unitary estimation, Ref.~\cite{chiribella05estimation} shows that the probe state $\ket{\phi}\in\H_{\bm{I}}\otimes\H_\text{aux}$ can be assumed without loss of generality to respect the identity $U^{\otimes k}\otimes \id \ket{\phi} = \id \otimes {U^T}^{\otimes k} \ket{\phi}$ for every $U\in\SU(d)$. Hence, every unitary estimation protocol can be converted to a parallel unitary transposition protocol with the same average fidelity.	
\end{proof}

	A direct implication of Thm.~\ref{theo:PAR=ESTtrans} is that any protocol for estimating $U^{\otimes k}$ can be converted to a parallel protocol for unitary transposition with the same average fidelity, \textit{vice versa}. Reference~\cite{acin00} shows that the optimal average fidelity for single copy unitary estimation is $\mean{F}_{\text{est},k=1}=  \frac{2}{d^2}$. Hence, for any dimension $d$ the optimal average fidelity for unitary transposition in the single copy ($k=1$) case is
\begin{equation}
	\mean{F}_{\text{trans},k=1}=  \frac{2}{d^2} \ .\label{eq:acin}
\end{equation}	
	Curiously, this the exactly same fidelity encountered in Ref.~\cite{ebler16} for optimal single copy unitary inversion. This can be easily explained in the qubit case, since there exists a deterministic exact protocol that transforms $U$ into $U^*$ which follows from the relation $\sigma_Y U \sigma_Y = U^*$, valid for every $U\in\SU(2)$. Hence, by combining transposition with complex conjugation, we obtain unitary inversion. For $d>2$, however, single copy deterministic unitary complex conjugation is not possible~\cite{miyazaki17}, hence any unitary inversion protocol which is a simple concatenation of transposition and conjugation is necessarily suboptimal. However, in Thm.~\ref{theo:PAR=ESTinv} we will show that when considering parallel strategies, the problem of unitary transposition is equivalent to unitary inversion for any dimension $d$ and any number of copies $k$.
	
	For the qubit case ($d=2$), Ref.~\cite{bagan03} proves that shown that the optimal average fidelity for qubit unitary estimation with $k$ copies is given by
\begin{equation} \label{eq:bagan}
	\mean{F}_{\text{est},d=2}=  \cos^2\left(\frac{\pi}{k+3}\right).
\end{equation} 
	Hence, due to Thm.~\ref{theo:PAR=ESTtrans}, this is also the optimal fidelity for qubit unitary transposition.

\subsubsection{Exponential advantage with sequential strategies}
	We now show that, for the task of deterministic unitary transposition, there exists an exponential gap between the performance of parallel and sequential strategies. 
	
\begin{theorem} \label{theo:SEQ>PAR}
For any dimension $d$, the optimal fidelity for parallel unitary transposition respects the upper bound 
\begin{equation}
	\mean{F}_\text{par}\leq 1-\frac{1}{(k+3)^2}.
\end{equation} 
Additionally, for any dimension $d$, there exists a sequential strategy attaining 
\begin{equation}
	\mean{F}_\text{seq}\geq  1-\left(1-\frac{1}{d^2}\right)^{\ceil{\frac{k}{d}}}.
\end{equation}
\end{theorem}
\begin{proof}
	We start with the observation that the group $\SU(2)$ is contained in $\SU(d)$. Consequently, the optimal average fidelity for qudit unitary estimation cannot be greater than qubit unitary estimation. Hence, due to Thm.~\ref{theo:PAR=ESTtrans}, for any dimension $d$, the optimal average fidelity for unitary transposition is bounded by  $\mean{F}= \cos^2\left( \frac{\pi}{k+3}\right)$. By expanding the cosine function in power series we see that 
\begin{align}
\cos^2\left(\frac{\pi}{k+3}\right)&\leq1 -\left(\frac{\pi}{k+3}\right)^2+\left(\frac{\pi}{k+3}\right)^4 \nonumber \\
&\leq 1 - \left(\frac{1}{k+3}\right)^2 \ .
\end{align}
	Hence, for any dimension, we have the upper bound  $\mean{F}_\text{par}\leq 1-\frac{1}{(k+3)^2}$.

For the sequential case, we use the probabilistic exact protocol presented in Ref.~\cite{quintino19PRA}) which has a success probability of $p_s=  1-\left(1-\frac{1}{d^2}\right)^{\ceil{\frac{k}{d}}}$, where $\ceil{.}$ stands for the ceiling function rounding to the next largest integer. As stated in Sec.~\ref{sec:prop}, every probabilistic exact protocol with $p_s$ leads to a deterministic non-exact protocol with average fidelity $\mean{F}\geq p_s$. Hence, we obtain the lower bound $\mean{F}_\text{seq}\geq  1-\left(1-\frac{1}{d^2}\right)^{\ceil{\frac{k}{d}}}$ and conclude the proof.
\end{proof}
	
	We finish this subsection by summarising the main results about optimal unitary transposition:
\begin{itemize}
\item (This work + Ref.~\cite{acin00}) For $k=1$ use, the optimal average fidelity for $U\mapsto U^T$ is $\mean{F} =\frac{2}{d^2}$
\item (This work + Ref.~\cite{bagan03}) The optimal average fidelity parallel strategy is the optimal fidelity for unitary dynamics estimation, and when $d=2$ we have $\mean{F}_\text{par}=\cos^2\left(\frac{\pi}{k+3}\right)$
\item (This work) The average fidelity for transforming $k$ uses of $d$-dimensional unitary operations respect $\mean{F}_\text{par}\leq 1-\frac{1}{(k+3)^2}$
\item (This work + Ref.~\cite{quintino19PRA}) There exists a sequential superchannel with $k$ uses that attains $\mean{F}_\text{seq}\geq 1-\left(1-\frac{1}{d^2}\right)^{\ceil{\frac{k}{d}}}$
\end{itemize}

\subsection{Unitary inversion}
	In this subsection, we discuss specific results for unitary inversion $f(U)=U^{-1}$, which is an anti-homomorphism since for unitary operations $U^{-1}=U^\dagger$ and $(UV)^\dagger=V^\dagger U^\dagger$. 

	The performance operator for unitary inversion is given by
\begin{align}
		\Omega=&\frac{1}{d^2} \int_\text{Haar} \dketbra{U^{-1}}{U^{-1}}_{PF}  \otimes \dketbra{U^*}{U^*}^{\otimes k}_{\bm{IO}} \;\dif U  \nonumber \\
		=&\frac{1}{d^2} \int_\text{Haar} \dketbra{U}{U}_{PF}  \otimes \dketbra{U^T}{U^T}^{\otimes k}_{\bm{IO}} \;\dif U,\label{eq:performance_inv}  
\end{align}
and respects the commutation relations
\begin{align} \label{eq:commutation_inv}
	[\Omega, A_P \otimes B_{\bm{I}}^{\otimes k} \otimes A_{\bm{O}}^{\otimes k}  \otimes B_F]&=0 \quad \forall A,B \in\SU(d)
\end{align}
	Following Eq.~\eqref{eq:perfomance_anti_homo}, the performance operator can be evaluated by finding a basis for operators commuting with $U^{\otimes (k+1)}$. In Appendix~\ref{app:explicit} we present an explicit form for the case $k=1$ and $k=2$.
	
	Similarly to the transposition case, parallel unitary inversion can always be implemented by a measure-and-prepare strategy and, it is also equivalent to the unitary estimation problem.

\begin{theorem} \label{theo:PAR=ESTinv}
	The optimal average fidelity for parallel unitary inversion can always be attained by measure-and-prepare strategies. Moreover, the value of the optimal average fidelity for parallel unitary transposition is exactly the value of the optimal average fidelity for estimating unitary operations uniformly sampled in $\SU(d)$. 
\end{theorem}
\begin{proof}
	Since unitary inversion is an anti-homomorphism, the first part of the proof of Thm.~\ref{theo:PAR=ESTtrans} ensures that, with no loss of generality, parallel superchannels for unitary inversion can be implemented by parallel measure-and-prepare strategies. 
	Let $S_\text{trans}$ be a parallel measure-and-prepare superchannel that attains unitary transposition.
	By definition of parallel measure-and-prepare superchannels $S_\text{trans}$  can be written as 
\begin{equation}
	S_\text{trans} = \sum_i (\rho*M_i)_{\bm{IO}} \otimes (R_i)_{PF}
\end{equation}
where $R_i$ are the Choi operator of quantum channels from $\H_P$ to $\H_F$. We can now define $S_\text{inv}:= \sum_i (\rho*M_i)_{\bm{IO}} \otimes (R_i^T)_{PF}$, which is a valid parallel measure-and-prepare superchannel because if $R_i$ is a quantum channel, $R_i^T$ is also a valid quantum channel.
From Lemma~\ref{theo:white_noise_cov}, with no loss of generality, we can write
\begin{align}
	S_\text{trans}*\dketbra{U}{U}^{\otimes k} & = \sum_i \tr\Big((\rho*M_i)\dketbra{U}{U}\Big) R_i 	\\
	&= \eta \dketbra{U^T}{U^T} + (1-\eta) \id \otimes \frac{\id}{d}
\end{align}
which holds true for every $U\in\SU(d)$.
	And since $\dketbra{U}{U}^T=\dketbra{U^*}{U^*}$, we can then see that the performance of unitary inversion is equivalent to unitary transposition
\begin{align}
S_\text{inv}*\dketbra{U}{U}^{\otimes k} & = \sum_i \tr\Big((\rho*M_i)\dketbra{U}{U}\Big) R_i^T 	\\
	&= \eta \dketbra{U^T}{U^T}^T + (1-\eta) \id \otimes \frac{\id}{d} \\
	&= \eta \dketbra{U^{-1}}{U^{-1}} + (1-\eta) \id \otimes \frac{\id}{d}.
\end{align}
	Hence, every parallel unitary transposition strategy can be mapped into a parallel unitary inversion strategy, \textit{vice versa}. Additionally, thanks to Thm.~\ref{theo:PAR=ESTtrans}, the problem of parallel unitary inversion is also equivalent to unitary estimation.
\end{proof}

\subsubsection{Exponential advantage with sequential strategies}
	We now prove that, for the task of deterministic unitary inversion, there exists an exponential gap between the performance of parallel and sequential strategies. 

\begin{theorem} \label{theo:SEQ>PAR_inv}
	The optimal fidelity for parallel unitary inversion respects the upper bound 
\begin{equation}
	\mean{F}_\text{par}\leq 1-\frac{1}{(k+3)^2}
\end{equation} 
and sequential strategies attain 
\begin{equation}
	\mean{F}_\text{seq}\geq  1-\left(1-\frac{1}{d^2}\right)^{\floor{\frac{k+1}{d}}}.
\end{equation}
\end{theorem}
\begin{proof}
	The proof follows the same steps of Thm.~\ref{theo:SEQ>PAR}. Since the group $\SU(2)$ is contained in $\SU(d)$, the optimal average fidelity for qudit unitary inversion cannot be greater than its qubit counterpart. Also, since $\sigma_YU^T\sigma_Y=U^{-1}$ for all $U\in\SU(2)$, qubit unitary inversion is equivalent to qubit unitary transposition. Hence, we have the upper bound $\mean{F}_\text{par}\leq 1-\frac{1}{(k+3)^2}$.

	For the sequential case, we use the probabilistic exact protocol presented in Ref.~\cite{quintino19PRL} which has a success probability of $p_s=  1-\left(1-\frac{1}{d^2}\right)^{\floor{\frac{k+1}{d}}}$, where $\floor{.}$ stands for the floor function, rounding down to the closest integer value. As discussed in Sec.~\ref{sec:prop}, every probabilistic exact protocol with $p_s$ leads to a deterministic non-exact protocol with average fidelity $\mean{F}\geq p_s$. Hence, we obtain the lower bound $\mean{F}_\text{seq}\geq 1-\left(1-\frac{1}{d^2}\right)^{\floor{\frac{k+1}{d}}}$ and conclude the proof.
\end{proof}

We finish this section by summarising the main results about optimal unitary inversion:
\begin{itemize}
\item (Ref\,\cite{ebler16}) For $k=1$ use, the optimal average fidelity for $U\mapsto U^{-1}$ is $\mean{F} =\frac{2}{d^2}$
\item (This work + Ref.~\cite{bagan03}) The optimal average fidelity parallel strategy is the optimal fidelity for unitary dynamics estimation, and when $d=2$ we have $\mean{F}_\text{par}=\cos^2\left(\frac{\pi}{k+3}\right)$
\item (This work) The average fidelity for transforming $k$ uses of $d$-dimensional unitary operations respect $\mean{F}_\text{par}\leq 1-\frac{1}{(k+3)^2}$
\item (This work + Ref.~\cite{quintino19PRL}) There exists a sequential superchannel with $k$ uses that attains $\mean{F}_\text{seq}\geq 1-\left(1-\frac{1}{d^2}\right)^{\floor{\frac{k+1}{d}}}$
\end{itemize}


\section{Computational results and the advantage of indefinite causality \label{sec:SDP2}}
	As discussed in Sec.~\ref{sec:SDP}, if the performance operator $\Omega:= \frac{1}{d^2}\int \dketbra{f(U)}{f(U)}_{PF} \otimes \dketbra{U^*}{U^*}^{\otimes k}_{\bm{IO}} \text{d} U$ corresponding to transforming $k$ copies of a unitary operator $U$ into $f(U)$ is given, the optimal strategy for such strategies can be phrased as semidefinite program. This is of particular interest as there exist efficient computational algorithms to solve SDPs~\cite{boyd}.
	
	In Appendix~\ref{app:explicit}, we have explicitly evaluated the performance operator for unitary transposition and unitary inversion for $k=1$ and $k=2$ and presented a general method that works for any $d$. With these performance operators and the SDP formulation presented in Eq.~\ref{eq:SDP}, we have used numerical packages to find the optimal performance for parallel, sequential, and adaptive strategies. We observe that all performance operators considered here respect $\Omega=\Omega^*$, hence, we can make more efficient computation by assuming, with no loss of generality, that the superchannels $S$ only have real numbers%
	\footnote{Indeed, let $S$ be a superchannel that attains an average fidelity $\mean{F}=\tr(S\Omega)$. In the computational basis, the matrix corresponding to the superchannel $S':=(S+S^*)/2$ only has real numbers and $\tr(S'\Omega)=\tr(S\Omega)$.}. %
	Our numerical results are presented in Fig.~\ref{Table:SDP_trans}. We have employed the SDP interpreter cvx~\cite{cvx} and performed independent evaluations with the solvers MOSEK~\cite{MOSEK}, and SDPT3~\cite{SDPT3}.

	Due to numerical floating-point precision, solutions arising from standard computational SDP solvers are not mathematical proofs. Nevertheless, it is possible to obtain rigorous upper and lower bounds for optimal average fidelity by means of computer assisted proofs. In this work, we have used the methods presented in the appendix ``III. COMPUTER-ASSISTED PROOFS'' of  Ref.~\cite{bavaresco21} to obtain rigorous upper and lower bounds for the optimal average fidelity. With this method, we could ensure that for $k=2$ copies, the values presented in Table~\ref{Table:SDP_trans} are valid up to the fourth decimal case. In a nutshell, the method consists in extracting the solution of the primal and dual of the SDP programming, truncating and use then as ansatz to obtain upper and lower bounds without making use of floating-point arithmetic.
	
	All our code can be found in the GitHub online repository \cite{MTQ_github_unitary_transformation} and can be freely used under the MIT license \cite{MIT_license}.
	\begin{figure}[h!] 
	\begin{center}
		\includegraphics[scale=0.49]{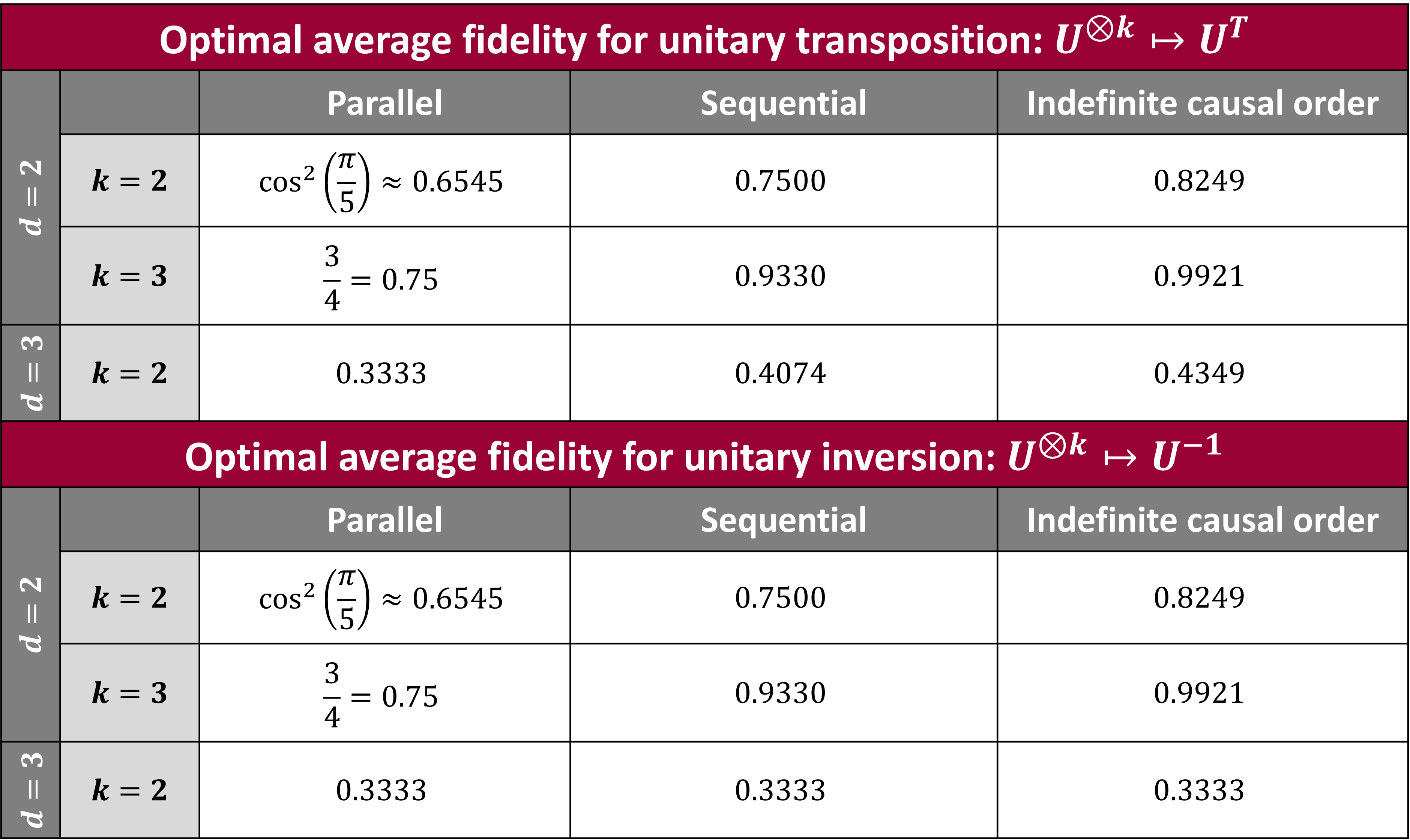} 
	\end{center}
\caption{Optimal average fidelity for deterministic protocols transforming $k$ uses of $U$ into its transpose (upper table) and into its inverse (lower table). The values for the parallel qubit case were proved analytically and the values for $k=2$ were obtained via numerical SDP optimisation and rigorously certified up to the fourth decimal digit with a computer assisted proof.}\label{Table:SDP_trans}
\end{figure}

\subsection{Advantages with indefinite causal order strategies}	\label{sec:advantage_indefinite}
	By analysing the results presented in Fig.~\ref{Table:SDP_trans} we see that for $d=2$, $k=2$ and $k=3$ general strategies outperform sequential ones for unitary transposition and unitary inversion. Also, for $d=3$ and $k=2$, general strategies outperform sequential ones for unitary transposition. In other words, indefinite causality is useful for deterministic unitary transposition and unitary inversion. Our computational approach provides an explicit description of general superchannels which outperform sequential strategies, but we could not find a simple pattern or intuition behind such processes. 
	We remark that, since the input operations we consider are identical unitary channels, superchannels of the quantum switch form \cite{chiribella09_switch} cannot offer an advantage over sequential strategies. Indeed, since the quantum switch transforms a pair of unitary operators $U_A$ and $U_B$ into the unitary coherent superposition 
	$\ketbra{0}{0}_c\otimes U_BU_A+\ketbra{1}{1}_c\otimes U_AU_B$,
where the subscript $c$ stands for control system. If $U_A=U_B=U$, the output of the quantum switch is a standard causally ordered circuit given by the unitary operator $\id_c\otimes UU$. Moreover, for tasks considering identical unitary channels as input operations, every $k$-slot switch-like superchannels\footnote{A two-slot switch-like superchannel transforms a pair of unitary operators $U_A$ and $U_B$ into the unitary coherent superposition
 $\ketbra{0}{0}_c\otimes V^{02}(U_B\otimes\id_\text{aux})V^{01}(U_A\otimes\id_\text{aux})V^{00}+\ketbra{1}{1}_c\otimes V^{12}(U_A\otimes\id_\text{aux})V^{11}(U_B\otimes\id_\text{aux})V^{10}$, where $V^{ij}$ are unitary operators and $\H_\text{aux}$ is an auxiliary space with arbitrary dimension. Note that if $V_{ij}=\id$ and the auxiliary spaces are trivial, we recover the quantum switch as a particular case.} can be simulated by causally ordered circuits, as shown in Thm.~4 of Ref.~\cite{bavaresco21b}.


\section{Port-based teleportation and quantum unitary transposition}
	Port-based teleportation (PBT) \cite{ishizaka08,ishizaka09} is a protocol which allows two distant parties to transfer quantum states, without the need of an active correction step (which is crucial in standard quantum teleportation schemes \cite{bennett93}). In the optimal deterministic PBT protocol, Alice and Bob share a resource state $\ket{\phi_\text{PBT}}\in \mathbb{C}_d^{\otimes 2k}$ in a way that Bob has access to $k$ different subsystems on his side (his share of the entangled state), which are entangled to the other $k$ subsystems on Alice's end. Each subsystem held by Bob is identified as a port to which Alice may transfer a target state $\ket{\psi}\in\mathbb{C}_d$. In order to teleport $\ket{\psi}$ to Bob, Alice performs a joint measurement with $k$ outcomes, each outcome corresponding to a port, on $\ket{\psi}$ and her part of the resource state $\ket{\phi_\text{PBT}}$. Alice then sends the outcome of her joint measurement to Bob, who finds a mixed state $\rho_{\text{out}}\approx \ketbra{\psi}{\psi}$ on the port informed by Alice. The performance of deterministic PBT is quantified in terms of the state fidelity between $\rho_{\text{out}}$, uniformly averaged between all pure target states $\ket{\psi}\in\mathbb{C}_d$. Also, the deterministic PBT may be viewed as a quantum channel $\map{C}$ which maps a target pure state $\ket{\psi}$ into $\rho_{\text{out}}\approx \ketbra{\psi}{\psi}$. Hence, its performance may be equivalently presented in terms of the fidelity between the identity map and $\map{C}$, a quantity usually referred to as entanglement fidelity. These two quantities are related via equation Eq.~\eqref{eq:fidelities}.
    
    The resource state of optimal PBT is known to respect the identity \cite{studzinski16}
\begin{equation}
    \left(U \otimes \id\right)^{\otimes k} \ket{\phi_\text{PBT}}
        =
    \left(\id \otimes U^T\right)^{\otimes k} \ket{\phi_\text{PBT}} \quad \quad \forall U\in\SU(d).
\end{equation}
    Hence, as noticed in Ref.~\cite{quintino19PRA}, PBT can be used to design a parallel strategy for unitary transposition. In particular, deterministic PBT leads to a parallel measure-and-prepare strategy, where the encoder step consists in preparing the state $\ket{\phi_\text{PBT}}\in \bigotimes_{i=1}^k
    \Big( \H_{I_i}\otimes\H_{\text{aux}_i}\Big) \cong \mathbb{C}_d^{\otimes 2k}$. Then, after applying $k$ parallel copies of an input unitary operator $U$, the system is described by $\left(U \otimes \id\right)^{\otimes k} \ket{\phi_\text{PBT}}\in \bigotimes_{i=1}^k
    \Big( \H_{O_i}\otimes\H_{\text{aux}_i}\Big)$. The decoder step consists in performing the PBT optimal joint measurement on the auxiliary space $\bigotimes_{i=1}^k\H_{\text{aux}_i}$ and the arbitrary input state $\ket{\psi}\in\H_P$, followed by discarding $k-1$ subsystems to obtain a state $\rho_\text{out} \approx U^T\ketbra{\psi}{\psi}\left(U^T\right)^\dagger$ in accordance to the measurement outcome. The average fidelity of unitary transposition for this protocol is then the optimal average entanglement fidelity for optimal deterministic PBT.
    
For qubits, the entangled fidelity obtained for optimal resource state deterministic PBT is \cite{ishizaka08}
    \begin{equation}
        \mean{F}_{\text{PBT},d=2}=\cos^2\left(\frac{\pi}{k+2}\right)
    \end{equation}
and in this work have shown that optimal parallel unitary transposition is given by
    \begin{equation}
        \mean{F}_{\text{par},d=2}=\cos^2\left(\frac{\pi}{k+3}\right).
    \end{equation}
We can then see that qubit unitary transposition superchannels adapted from port-based teleportation have a strictly smaller average fidelity  than the optimal qubit parallel unitary transposition protocol. One possible explanation for this difference is the fact that in PBT, the decoder step is after the joint measurement corresponds to PBT correction and consists in simply discarding $k-1$ subsystems, while for optimal parallel unitary transposition we allow general correction operations.
    
    We remark that, when considering the task of probabilistic exact unitary transformations, Ref.~\cite{quintino19PRA} showed that the maximal success probability can always be obtained by probabilistic PBT%
    \footnote{Probabilistic unitary transposition protocols based on PBT attain optimal success probability, however these PBT based protocols require a large dimension for auxiliary space. It is possible to attain the optimal parallel unitary transposition with considerably smaller auxiliary spaces with a different parallel protocol, see Ref.~\cite{sedlak18}.}.%
This points out another difference between the probabilistic exact and deterministic transformations.
	
\section{Discussions}
Understanding quantum operations and transformations thereof is an essential step in the quest of developing novel quantum technology. On the one hand, said devices are currently prone to limitations stemming from noisy hardware and low computational power -- such that implementing a certain transformation in the most efficient way is crucial. On the other hand, fundamental limitations, such as no-go theorems in quantum theory, need to be addressed in order to construct general-purpose quantum devices.

In this work, we considered the task of transforming multiple uses of an unknown unitary transformation into a target operation. Concretely, we studied how the composition of the uses can affect the run-time and quality of the transformation. We showed that, depending on the structural property of the transformation, sequential application of the unitary may dramatically outperform parallel ones. This contrasts previous insights \cite{bisio13}, where parallel uses of channels were identified to be the optimal arrangement.

As part of our methods, we established a one-to-one connection between estimating unitary operations, parallel circuits for the task of transposing an unknown unitary operation, and parallel circuits for the task of inverting an unknown unitary operation. This allowed us to apply results from one task to the other, yielding new insights for transformations which do not belong to the class studied previously in the literature \cite{bisio13}. We further show that arrangements of channels without definite causal order may outperform sequential circuits. These processes are shown not to be of the quantum switch form \cite{chiribella09_switch} -- motivating further investigation of indefinite causality beyond the quantum switch.

We further studied the task of deterministic unitary inversion through multiple uses of an unknown input operation.  In the single use regime, deterministic unitary inversion was studied in Ref.~\cite{ebler16}. When multiple uses are considered, probabilistic heralded inversion of unitary operations is investigated in Ref.~\cite{quintino19PRL}, which considers probabilistic circuits and probabilistic superchannels. Additionally, it was recently shown that, under reasonable conditions, a function $f:\SU(d)\to\SU(d)$ is quantum time reversal if and only if $f$ it is equivalent to unitary transposition or unitary inversion \cite{chiribella20timeflip}, two classes of transformations covered by our methods.    We also remark Ref.~\cite{navascues17}, which provides a probabilistic heralded method to reset and invert unitary operations even in absence of control on the target system in a scheme which was extended and optimised in Ref.~\cite{trillo19}.

An interesting future research direction would be the distillation of operation from others, beyond transformations acting as homomorphisms and anti-homomorphisms. Such direction could, for instance, use a given set of input channels and fixed number of uses thereof as a resource to mimic the action of a different one.  This line of research would be leading into the direction of transfer learning for transformations and algorithms, by using readily implementable routines to approximate more complex ones.
	

\section*{Acknowledgements}
	We thank Qingxiuxiong Dong,  Mio Murao, Atsushi Shimbo, Akihito Soeda, Micha\l\, Studzinski for fruitful discussions and Tomasz Młynik for noting a typo in an earlier version of the manuscript. We acknowledge Q-leap (MEXT Q-leap JPMXS0118069605) and the JSPS Kakenhi 18H04286.
	M.T.Q. acknowledges the Austrian Science Fund (FWF) through the SFB project BeyondC (sub-project F7103), a grant from the Foundational Questions Institute (FQXi) as part of the  Quantum Information Structure of Spacetime (QISS) Project (qiss.fr). The opinions expressed in this publication are those of the authors and do not necessarily reflect the views of the John Templeton Foundation. This project has received funding from the European Union’s Horizon 2020 research and innovation programme under the Marie Skłodowska-Curie grant agreement No 801110.
It reflects only the authors' view, the EU Agency is not responsible for any use that may be made of the information it contains. ESQ has received funding from the Austrian Federal Ministry of Education, Science and Research (BMBWF)


\clearpage
\appendix
\section*{Appendix}
\section{Additional details for the homomorphic case: $f(UV)=f(U)f(V)$} \label{app:f_is_rep}
\subsection{Evaluating the performance operator $\Omega$}
	In this subsection, we present one method for evaluating the performance operator 
\begin{equation}
	\Omega:= \frac{1}{d^2}\int_{\text{Haar}} \dketbra{f(U)}{f(U)}_{PF} \otimes \dketbra{U^*}{U^*}^{\otimes k}_{\bm{IO}} \text{d} U
\end{equation}
	when $f:\SU(d)\to\SU(d')$ is a homomorphism. As shown in the main text, the performance operator $\Omega$ respects the commutation relations
\begin{align}
	[\Omega, \id_P \otimes \id_{\bm{I}} \otimes U_{\bm{O}}^{*^{\otimes k}} \otimes f(U)_F]&=0, \quad \forall U\in\SU(d) 
\end{align}
	Note that linear operators $P\in \L(\H_{\bm{O}} \otimes \H_F)$ respecting $[P,U_{\bm{O}}^{*^{\otimes k}} \otimes f(U)_F]=0$  form a linear subspace, hence there exists an orthogonal basis of operators $\{P^i\}_i$, 
	$\tr(P^iP^j)=\delta_{ij} d_i$, such that $P$ can be written as $P=\sum_i \alpha^i P^i$ for some complex coefficients $\alpha^i$. 
	Using the orthogonality relation $\tr(P{^i}^\dagger P^j)=d_i \delta_{ij}$, the performance operator $\Omega$ can be written as 
\begin{equation}
	\Omega = \sum_i \Omega^i_{P\bm{I}} \otimes P^i_{\bm{O}F}
\end{equation}
for some linear operators $\Omega^i_{P\bm{I}}$. 
	We can obtain $\Omega^i_{P\bm{I}} $ explicitly, by observing that, since  $\{P^i\}_i$ is an orthonormal basis, we have 
	\small
\begin{align}
	\Omega^i_{P\bm{I}}\,d_i &=
\tr_{\bm{O}F} \left( \left(\id_{P\bm{I}}\otimes  P_{\bm{O}F}^{{i}^\dagger} \right) \Omega \right) \nonumber \\
&= \frac{1}{d^2}\int
\tr_{\bm{O}F} \Big( \left( \id_{P\bm{I}} \otimes P_{\bm{O}F}^{{i}^\dagger} \right) \dketbra{f(U)}{f(U)}_{PF} \otimes \dketbra{U^*}{U^*}^{\otimes k}_{\bm{IO}} \Big) \text{d} U \nonumber \\
&= \frac{1}{d^2}\int
\tr_{\bm{O}F} \Big(  \id_{P\bm{I}} \otimes \left[ \left(f(U)\otimes {U^*}^{\otimes k}\right)^\dagger
 P^{{i}^\dagger} f(U)\otimes U^{*^{\otimes k}} \right]_{\bm{O}F} \dketbra{\id}{\id}_{PF} \otimes \dketbra{\id}{\id}^{\otimes k}_{\bm{IO}} \Big) \text{d} U \nonumber \\
 &= \frac{1}{d^2}\int
\tr_{\bm{O}F} \Big(  \id_{P\bm{I}} \otimes \left[ 
 P^{{i}^\dagger} \left(f(U)\otimes U^{*^{\otimes k}}\right)^\dagger f(U)\otimes U^{*^{\otimes k}} \right]_{\bm{O}F} \dketbra{\id}{\id}_{PF} \otimes \dketbra{\id}{\id}^{\otimes k}_{\bm{IO}} \Big) \text{d} U \nonumber \\
  &= \frac{1}{d^2}\int
\tr_{\bm{O}F} \Big(  \id_{P\bm{I}} \otimes \left[ 
 P^{{i}^\dagger} \right]_{\bm{O}F} \dketbra{\id}{\id}_{PF} \otimes \dketbra{\id}{\id}^{\otimes k}_{\bm{IO}} \Big) \text{d} U \nonumber \\
 &= \frac{1}{d^2}
\tr_{\bm{O}F} \Big(  \id_{P\bm{I}} \otimes \left[ 
 P^{{i}^\dagger} \right]_{\bm{O}F} \dketbra{\id}{\id}_{PF} \otimes \dketbra{\id}{\id}^{\otimes k}_{\bm{IO}} \Big) \nonumber \\
  &= \frac{1}{d^2}
\tr_{\bm{O}F} \Big( \left[  P^{{i}^*}_{\bm{I}P}\otimes \id_{\bm{O}F} 
 \right] \dketbra{\id}{\id}_{PF}  \dketbra{\id}{\id}^{\otimes k}_{\bm{IO}} \Big) \nonumber \\
   &= \frac{1}{d^2}
P^{{i}^*}_{\bm{I}P}.
\end{align}
\normalsize
	Hence, we can write
\begin{equation}
	\Omega = \frac{1}{d^2}\sum_i \frac{\left(P^i_{\bm{I}P}\right)^* \otimes P^i_{\bm{O}F}}{d_i}.
\end{equation}
	In Appendix~\ref{app:explicit} we present an explicit basis $\{P^i\}_i$ for the case $f(U)=U^*$.

\subsection{Optimality of parallel strategies for the homomorphic case} \label{app:par_is_optimal}

	This subsection is dedicated to prove Proposition~\ref{prop:f_is_rep} stated in the main text. We start with a lemma which first appeared in Ref.~\cite{bisio13} which here we present it and prove it here in a slightly different manner.

\begin{lemma}[Proposition 4 of Ref.~\cite{bisio13}] \label{lemma:bisio}
	Let
$S\in\L(\H_P\otimes\H_{\bm{I}}\otimes\H_{\bm{O}}\otimes\H_F)$ be a general $k$-slot superchannel. If the operator $H:=\tr_F(W)\in\L(\H_P\otimes\H_{\bm{I}}\otimes\H_{\bm{O}})$ respects the commutation relation 
\begin{equation}
	H \, \left( \id_{P}\otimes\id_{\bm{I}} \otimes {U_{\bm{O}}^{\otimes k}}\right)  = 
	\left(\id_{P}\otimes\id_{\bm{I}} \otimes {U_{\bm{O}}^{\otimes k}} \right)\, H ,
\end{equation}
for every unitary operators from a set $\{U\}_U$, $U\in\L(\mathbb{C}_d)$. There exists a parallel $k$-slot superchannel $S'$ such that 
\begin{equation}
S'*\dketbra{U^{\otimes k}}{U^{\otimes k}}_{\bm{IO}} =
	S*\dketbra{U^{\otimes k}}{U^{\otimes k}}_{\bm{IO}} 	
\end{equation}
for every unitary operators from a set $\{U\}_U$.

Moreover, if we set $\H_{P'}$ is an auxiliary space which is isomorphic to $\H_P$, $\H_{\bm{I}'}$ is an auxiliary space which is isomorphic to $\H_{\bm{I}}$,. The parallel superchannel $S'$ can be written as $S'=E*D$ where $E$ is the Choi operator of a quantum channel from $\L(\H_P)$ to $\L(\H_{P'}\otimes\H_{\bm{I}}\otimes\H_{\bm{I}'})$ defined by
 \begin{align}
	E_{PP'\bm{I}\bm{I}'} := \left(\sqrt{H}_{P'\bm{I}'\bm{I}}\otimes\id_P\right)
	 \dketbra{\id}{\id}_{\bm{I}'\bm{I}}\otimes\dketbra{\id}{\id}_{P'P}
	 \left(\sqrt{H}_{P'\bm{I}'\bm{I}}\otimes\id_P\right)
 \end{align}
and $D$ is a quantum channel from $\L(\H_{P'}\otimes\H_{\bm{I}'}\otimes\H_{O})$ to $\L(\H_F))$ defined by
\begin{equation}
	D_{P'\bm{I}'\bm{O}F} := \left(\sqrt{H}_{P'\bm{I}'\bm{O}}^{-1}\otimes\id_F \right)^T
	\, S_{P'\bm{I}'\bm{O}F}\,\left(\sqrt{H}_{P'\bm{I}'\bm{O}}^{-1}\otimes\id_F \right)^T
\end{equation}
\end{lemma}    

\begin{proof}
	We start our proof by verifying that $E$ is a valid quantum channel from $\L(\H_P)$ to $\L(\H_{P'\bm{I}'\bm{I}})$. The operator $E$ is positive semidefinite because it is a composition of positive semidefinite operators and the normalisation condition follows from $\tr_{P'\bm{I}\bm{I}'}(E)=\id_P $. We now ensure that The operator $D$ corresponds to a valid quantum channel from $\L(\H_{P'\bm{I}'\bm{O}})$ to $\L(\H_{F})$. The operator $D$ is positive semidefinite because it is a composition of positive semidefinite operators, and it is also direct to check that $\tr_F(D)=\id_{P'\bm{I}'\bm{I}}$.

Our next step starts by pointing out that $H$ commutes with $\id\otimes\id\otimes U^{\otimes k}$. Hence,
$\sqrt{H}$ also commutes with $\id\otimes\id\otimes U^{\otimes k}$. Consequently, we can evaluate the link product%
\footnote{We recall that if $C_{\bm{IO}}$ is the Choi operator of a map $\map{C}:\L(\H_{\bm{I}})\to\L(\H_{\bm{O}})$, for every linear operator and $\rho\in\L(\H_A\otimes\H_{\bm{I}})$ we have that $C*\rho=\left[\map{\id}\otimes\map{C}(\rho)\right]_{AO}$. }
\begin{align}
	E_{P'\bm{I}'\bm{I}P}*\dketbra{U^{\otimes k}}{U^{\otimes k}}_{\bm{IO}} =&
	U^{\otimes k}_{\bm{O}} \left[ E_{P'\bm{I}'\bm{O}P}\right] U^{\otimes k}_{\bm{O}} \nonumber \\
=& U^{\otimes k}_{\bm{O}} 
\left[\left(\sqrt{H}_{P'\bm{I}'\bm{O}}\otimes\id_P\right)
	 \dketbra{\id}{\id}_{\bm{I}'\bm{O}}\otimes\dketbra{\id}{\id}_{P'P}
	 \left(\sqrt{H}_{P'\bm{I}'\bm{O}}\otimes\id_P\right)\right]
	  U^{\otimes k}_{\bm{O}} \nonumber \\
=&\left(\sqrt{H}_{P'\bm{I}'\bm{O}}\otimes\id_P\right)
\dketbra{U^{\otimes k}}{U^{\otimes k}}_{\bm{I}'\bm{O}}\otimes\dketbra{\id}{\id}_{P'P}
\left(\sqrt{H}_{P'\bm{I}'\bm{O}}\otimes\id_P\right).
\end{align}
We now finish the proof by checking that
\small
\begin{align}
S'*\dketbra{U^{\otimes k}}{U^{\otimes k}}
=&E_{P'\bm{I}'\bm{I}P}*D_{P'\bm{I}'\bm{O}F}*\dketbra{U^{\otimes k}}{U^{\otimes k}}_{\bm{IO}} \nonumber \\
=&E_{P'\bm{I}'\bm{I}P}*\dketbra{U^{\otimes k}}{U^{\otimes k}}_{\bm{IO}}*D_{P'\bm{I}'\bm{O}F} \nonumber \\
=&\left(\sqrt{H}_{P'\bm{I}'\bm{O}}\otimes\id_P\right)
\dketbra{U^{\otimes k}}{U^{\otimes k}}_{\bm{I}'\bm{O}}\otimes\dketbra{\id}{\id}_{P'P}
\left(\sqrt{H}_{P'\bm{I}'\bm{O}}\otimes\id_P\right)*D_{P'\bm{I}'\bm{O}F} \nonumber \\
=&\tr_{P'\bm{I}'\bm{O}}\left(\sqrt{H}_{P'\bm{I}'\bm{O}} 
\dketbra{U^{\otimes k}}{U^{\otimes k}}_{\bm{I}'\bm{O}}\otimes\dketbra{\id}{\id}_{P'P}
\sqrt{H}_{P'\bm{I}'\bm{O}} D_{P'\bm{I}'\bm{O}F}^{T_{P'\bm{I}'\bm{O}}} \right)\nonumber \\
=&\tr_{P'\bm{I}'\bm{O}}\left(\sqrt{H}_{P'\bm{I}'\bm{O}} 
\dketbra{U^{\otimes k}}{U^{\otimes k}}_{\bm{I}'\bm{O}}\otimes\dketbra{\id}{\id}_{P'P}
\sqrt{H}_{P'\bm{I}'\bm{O}} 
\,\sqrt{H}_{P'\bm{I}'\bm{O}}^{-1}
	\, S_{P'\bm{I}'\bm{O}F}^{T_{P'\bm{I}'\bm{O}}}\,\sqrt{H}_{P'\bm{I}'\bm{O}}^{-1}\right)\nonumber \\
=&\tr_{P'\bm{I}'\bm{O}}\left(\dketbra{U^{\otimes k}}{U^{\otimes k}}_{\bm{I}'\bm{O}}\otimes\dketbra{\id}{\id}_{P'P}
S_{P'\bm{I}'\bm{O}F}^{T_{P'\bm{I}'\bm{O}}} \right) \nonumber \\
=&\tr_{P'\bm{I}'\bm{O}}\left(\dketbra{U^{\otimes k}}{U^{\otimes k}}_{\bm{I}'\bm{O}}\otimes\dketbra{\id}{\id}_{P'P}
S_{P\bm{I}'\bm{O}F}^{T_{\bm{I}'\bm{O}}} \right) \nonumber \\
=&\tr_{\bm{I}'\bm{O}}\left(\dketbra{U^{\otimes k}}{U^{\otimes k}}_{\bm{I}'\bm{O}}
S_{P\bm{I}'\bm{O}F}^{T_{\bm{I}'\bm{O}}} \right) \nonumber \\
=&\dketbra{U^{\otimes k}}{U^{\otimes k}}_{\bm{I}'\bm{O}}
*S_{P\bm{I}'\bm{O}F} \nonumber \\
=&S_{P\bm{IO}F}*\dketbra{U^{\otimes k}}{U^{\otimes k}}_{\bm{IO}}.
\end{align}
\normalsize
\end{proof}

	We are now in conditions to prove Proposition~\ref{prop:f_is_rep}.
	
\begin{proof}
	Since the performance operator respects the commutation relations of Eq.~\eqref{eq:f(u)_commutation} and Eq.~\eqref{eq:f(u)_commutation2}, for any superchannel $S$, we can define the ``Haar-twirled'' operator 
\begin{equation}
	S':=\int\hspace*{-2mm}\int \Big(f(A)_P \otimes A_{\bm{I}}^{*^{\otimes k}} \otimes B_{\bm{O}}^{*^{\otimes k}} \otimes f(B)_F\Big)
		S
		\Big( f(A)_P \otimes A_{\bm{I}}^{*^{\otimes k}} \otimes B_{\bm{O}}^{*^{\otimes k}} \otimes f(B)_F \Big)^\dagger  \dif A\dif B
\end{equation}	
	 which respects
\begin{equation} \label{eq:commutation_S_rep}
	[S', f(A)_P \otimes A_{\bm{I}}^{*^{\otimes k}} \otimes B_{\bm{O}}^{*^{\otimes k}} \otimes f(B)_F]=0,
\end{equation}
for any unitary operators $A,B$ and the average fidelity
\begin{equation}
	\tr(S\Omega) = \tr(S'\Omega).
\end{equation}
The Eq.~\eqref{eq:commutation_S_rep} implies that $S'$ is respects the covariant relation,
\footnotesize
	\begin{align}
		S'*\dketbra{U}{U}^{\otimes k} &=\Bigg(\Big(f(A)_P \otimes A_{\bm{O}}^{*^{\otimes k}} \otimes B_{\bm{O}}^{*^{\otimes k}} \otimes f(B)_F\Big)
		S'
		\Big( f(A)_P \otimes A_{\bm{I}}^{*^{\otimes k}} \otimes B_{\bm{O}}^{*^{\otimes k}} \otimes f(B)_F \Big)^\dagger \Bigg)
		*\dketbra{U}{U}^{\otimes k}\\
	&=\Big( f(A)_P \otimes f(B)_F\Big) 
\Bigg(S' * \bigg(
A_{\bm{I}}^{\dagger^{\otimes k}} \otimes B_{\bm{O}}^{\dagger^{\otimes k}} \dketbra{U}{U}^{\otimes k} A_{\bm{I}}^{{\otimes k}} \otimes B_{\bm{O}}^{{\otimes k}}  \bigg)  \Bigg)
		\Big(f(A)_P \otimes f(B)_F\Big)^\dagger \nonumber \\
&=\Big( f(A)_P \otimes f(B)_F\Big) 
\Bigg(S' * \bigg( \dketbra{B^\dagger UA^*}{B^\dagger UA^*}^{\otimes k}\bigg)  \Bigg)
		\Big(f(A)_P \otimes f(B)_F\Big)^\dagger 
	\end{align}
\normalsize
Hence, if we set $A=\id$ and $B=U$ we prove Eq.~\eqref{eq:S_is_covariant}.

	We now show that $S'$ is a valid general superchannel. For that, it is enough to show that for any non-signalling channel $C_{\bm{IO}}$ we have that $S'*C_{\bm{IO}}$ is a valid channel. Note that if $C_{\bm{IO}}$ is a non-signalling channel, 
\begin{equation}
C'_{\bm{IO}}:=\int \Bigg(\Big( A^{*{\otimes k}}_{\bm{I}} \otimes B^{*{\otimes k}}_{\bm{O}}\Big)^\dagger C_{\bm{IO}}^T \Big(A^{*{\otimes k}}_{\bm{I}} \otimes B^{*{\otimes k}}_{\bm{O}}\Big)\Bigg)\dif B
\end{equation} is also a non-signalling channel. We can then show that $S'*C_{\bm{IO}}$ is a valid channel by direct calculation:
\footnotesize
\begin{align}
\tr_F\Big(S'*C_{\bm{IO}}\Big) &= \tr_{\bm{IO}F} \Big( S' \; \id_P\otimes C_{\bm{IO}}^T\otimes \id_F\Big) \nonumber \\
&=\tr_{\bm{IO}F}\Bigg\{\int\hspace*{-2mm}\int \Big(f(A)_P \otimes A_{\bm{O}}^{*^{\otimes k}} \otimes B_{\bm{O}}^{*^{\otimes k}} \otimes f(B)_F\Big)
		S \cdot \nonumber \\
		& \qquad \qquad \qquad \qquad \qquad \qquad \qquad \cdot
		\Big( f(A)_P \otimes A_{\bm{I}}^{*^{\otimes k}} \otimes B_{\bm{O}}^{*^{\otimes k}} \otimes f(B)_F \Big)^\dagger	
		\; \id_P \otimes C_{\bm{IO}}^T \otimes \id_F	 \Bigg\}
	\dif A\dif B \nonumber \\
		&=\tr_{\bm{IO}F}\Bigg\{\int\hspace*{-2mm}\int \Big(f(A)_P \otimes A_{\bm{I}}^{*^{\otimes k}} \otimes B_{\bm{O}}^{*^{\otimes k}} \otimes \id_F\Big)
		S \cdot \nonumber \\
		& \qquad \qquad \qquad \qquad \qquad \qquad \qquad \cdot
		\Big( f(A)_P \otimes A_{\bm{I}}^{*^{\otimes k}} \otimes B_{\bm{O}}^{*^{\otimes k}} \otimes \id_F \Big)^\dagger	
		\; \id_P \otimes C_{\bm{IO}}^T \otimes \id_F	 \Bigg\}
	\dif A\dif B 	\nonumber \\
		&=\tr_{\bm{IO}F}\Bigg\{\int\hspace*{-2mm}\int \Big(f(A)_P \otimes \id_{\bm{I}}  \otimes \id_{\bm{O}} \otimes \id_F\Big)
		S \cdot \nonumber \\
		& \qquad \qquad \qquad \qquad \qquad \qquad \qquad \cdot
		\Big( f(A)_P \otimes \id_{\bm{I}} \otimes \id_{\bm{O}} \otimes \id_F \Big)^\dagger	
		\; \id_P \otimes C'_{\bm{IO}} \otimes \id_F	 \Bigg\}
	\dif A\dif B 	\nonumber \\
	&= \int\Big( f(A)_P \otimes \id_F\Big) \tr_{\bm{IO}F}\Big(S\; \id_P\otimes {C'}_{\bm{IO}}\otimes \id_F \Big) \Big(f(A)_P \otimes \id_F\Big)^\dagger \dif A \nonumber  \\
	&=\id_P.
\end{align}
\normalsize
	
	We finish the proof by showing that $S'$ is a parallel superchannel. To attain this goal, we first verity that the partial trace of $S'$ respects the commutation relation
	\begin{equation}
		\Big[\tr_F(S'),\, \id_P \otimes \id_{\bm{I}}^{\otimes k} \otimes B_{\bm{O}}^{*^{\otimes k}}\Big]=0
	\end{equation}
	for any unitary operator $B$. Hence, the superchannel $S'$ satisfies the hypothesis of Lemma\,\ref{lemma:bisio}.
\end{proof}
 
\section{Additional details for the anti-homomorphic case: $f(UV)=f(V)f(U)$} \label{app:fT_is_rep}
	When $f:\SU(d)\to\SU(d')$ is an anti-homomorphism, the performance operator
\begin{equation}
	\Omega:= \frac{1}{d^2}\int_{\text{Haar}} \dketbra{f(U)}{f(U)}_{PF} \otimes \dketbra{U^*}{U^*}^{\otimes k}_{\bm{IO}} \text{d} U
\end{equation} can be explicitly evaluated following steps analogous  to the case where $f$ is a homomorphism (see Appendix~\ref{app:f_is_rep}). In this case, we obtain 
\begin{equation}
	\Omega = \frac{1}{d^2}\sum_i \frac{\left(P_{\bm{O}P}^i\right)^* \otimes P^i_{\bm{I}F}}{d_i},
\end{equation}
	where $\{P_i\}_i$ is an orthonormal basis for the set of operators $P$ respecting $[P,U^{*^{\otimes k}}\otimes f(U)^T]=0$ for all $U\in\SU(d)$.
	In Appendix~\ref{app:explicit} we present an explicit basis $\{P^i\}_i$ for the case $f(U)=U^T$ and $f(U)=U^{-1}$.

	We now re-state and prove Lemma~\ref{lemma:fT_is_rep} from the main text. 
\setcounter{lemma}{0}
\begin{lemma} 
	Let $S\in\L(\H_P\otimes\H_{\bm{I}}\otimes\H_{\bm{O}}\otimes\H_F)$ be a parallel/sequential/general superchannel that transforms $k$ uses of a unitary operator $U$ into $f(U)$ with average fidelity $\mean{F}$. If $f(UV)=f(V)f(U)$, there exists a parallel/sequential/general superchannel $S'\in\L(\H_P\otimes\H_{\bm{I}}\otimes\H_{\bm{O}}\otimes\H_F)$ that transforms $k$ uses of a unitary operator $U$ into $f(U)$ with average fidelity $\mean{F}$ and respects the commutation relation
\begin{equation}
	\Big[S', f(A)_P^T \otimes B_{\bm{I}}^{*^{\otimes k}} \otimes A_{\bm{O}}^{*^{\otimes k}} \otimes f(B)^T_F\Big]=0,
\end{equation}
for any unitary operators $A,B$.

	In particular, the action of $S'$ on $\dketbra{U}{U}^{\otimes k}$ is described by its action on the identity channel $\dketbra{\id}{\id}^{\otimes k}$ via
\begin{equation} 
	S'*\dketbra{U}{U}^{\otimes k} = \Big(\id_P \otimes {f(U^T)}^T_F\Big)
	 \Big( S'*\dketbra{\id}{\id}^{\otimes k}_{\bm{IO}}\Big)
	 \Big(\id_P \otimes {f(U^T)}^T_F\Big)^\dagger .
\end{equation}
\end{lemma}

\begin{proof}
	Since the performance operator respects the commutation relations of Eq.~\eqref{eq:f(u)T_commutation} and Eq.~\eqref{eq:f(u)T_commutation2}, for any superchannel $S$, we can define the ``Haar-twirled'' operator 
\begin{equation}
	S':=\int\hspace*{-2mm}\int \Big(f(A)_P^T \otimes B_{\bm{I}}^{*^{\otimes k}} \otimes A_{\bm{O}}^{*^{\otimes k}} \otimes f(B)^T_F\Big)
		S
		\Big( f(A)_P^T \otimes B_{\bm{I}}^{*^{\otimes k}} \otimes A_{\bm{O}}^{*^{\otimes k}} \otimes f(B)^T_F \Big)^\dagger  \dif A\dif B
\end{equation}	
	 which respects
\begin{equation} \label{eq:commutation_S_trans_rep}
	\Big[S', f(A)_P^T \otimes B_{\bm{I}}^{*^{\otimes k}} \otimes A_{\bm{O}}^{*^{\otimes k}} \otimes f(B)^T_F\Big]=0,
\end{equation}
for any unitary operators $A,B$ and
\begin{equation}
	\tr(S\Omega) = \tr(S'\Omega).
\end{equation}
The Eq.~\eqref{eq:commutation_S_trans_rep} implies that $S'$ is respects the covariant relation,
\footnotesize
	\begin{align}
S'*\dketbra{U}{U}&=\Bigg(\Big(f(A)_P^T \otimes B_{\bm{I}}^{*^{\otimes k}} \otimes A_{\bm{O}}^{*^{\otimes k}} \otimes f(B)^T_F\Big)
		S'
\Big( f(A)_P^T \otimes B_{\bm{I}}^{*^{\otimes k}} \otimes A_{\bm{O}}^{*^{\otimes k}} \otimes f(B)^T_F\Big )^\dagger \Bigg)
		*\dketbra{U}{U} \nonumber \\
&=\Big( f(A)^T_P \otimes f(B)^T_F\Big) 
\Bigg(S' * \bigg(
B_{\bm{I}}^{\dagger^{\otimes k}} \otimes A_{\bm{O}}^{\dagger^{\otimes k}} \dketbra{U}{U}^{\otimes k} B_{\bm{I}}^{{\otimes k}} \otimes A_{\bm{O}}^{{\otimes k}}  \bigg)  \Bigg)
\Big(f(A)^T_P \otimes f(B)^T_F\Big)^\dagger \nonumber \\
&=\Big( f(A)^T_P \otimes f(B)^T_F\Big) 
\Bigg(S' * \bigg( \dketbra{A^\dagger UB^*}{A^\dagger UB^*}^{\otimes k}\bigg)  \Bigg)
	\Big( f(A)^T_P \otimes f(B)^T_F\Big) ^\dagger 			
	\end{align}
\normalsize
Hence, if we set $A=\id$ and $B=U^T$ we obtain Eq.~\eqref{eq:S_is_covariant2}.

	Following the same steps used in the proof of Proposition\,\ref{prop:f_is_rep}, we can show that if $S$ is a general superchannel, $S'$ is a general superchannel. Also, if $S$ is a parallel/sequential superchannel, $S'$ is a parallel/sequential superchannel.
\end{proof}


\section{Proof of Thm.~\ref{thm:WCF}}
	We now re-state and prove Thm.~\ref{thm:WCF} from the main text.
\setcounter{theorem}{0}
\begin{theorem}
Let $S\in\L(\H_P\otimes\H_{\bm{I}}\otimes\H_{\bm{O}}\otimes\H_F)$ be a parallel/sequential/general superchannel that transforms $k$ uses of a unitary operator $U\in\SU(d)$ into $f(U)\in\SU(d')$ with average fidelity $\mean{F}$. If $f$ is a homomorphism, \ie, $f(UV)=f(U)f(V)$, or an anti-homomorphism, \ie, $f(UV)=f(V)f(U)$, there exists a parallel/sequential/general superchannel $S'\in\L(\H_P\otimes\H_{\bm{I}}\otimes\H_{\bm{O}}\otimes\H_F)$ that transforms $k$ uses of a unitary operator $U$ into $f(U)$ with worst-case fidelity $F_{\rm wc}=\mean{F}$.
\end{theorem}

\begin{proof}
	Proposition\,\ref{prop:f_is_rep} shows that when $f$ is a representation of a group and $S$ is a parallel/sequential/general superchannel that transforms $k$ uses of $U$ into $f(U)$ with average fidelity $\mean{F}$, there exists a parallel/sequential/general superchannel $S'$ that transforms $k$ uses of $U$ into $f(U)$ with average fidelity $\mean{F}$ and respects
\begin{equation}
	S'*\dketbra{U}{U}= \id_P\otimes f(U)_F \Big( S'*\dketbra{\id}{\id}_{\bm{IO}}^{\otimes k} \Big) \id_P \otimes f(U)_F^\dagger
\end{equation}
for every unitary operator $U\in\SU(d)$.
Hence, for the covariant superchannel $S'$, the fidelity between $S'*\dketbra{U}{U}^{\otimes k}$ and $\dketbra{f(U)}{f(U)}$ is independent of $U$ and the average coincides with the worst-case fidelity, that is
\begin{align}
F\Big(S'*\dketbra{U}{U}^{\otimes k},\dketbra{f(U)}{f(U)}\Big)&=
\dbra{f(U)}\bigg(\id_P\otimes f(U)_F \Big( S'*\dketbra{\id}{\id}_{\bm{IO}}^{\otimes k} \Big) \id_P\otimes f(U)^\dagger_F\bigg) \dket{f(U)} \nonumber \\
&=
\dbra{\id} \Big( S'*\dketbra{\id}{\id}_{\bm{IO}}^{\otimes k} \Big) \dket{\id} \nonumber  \\
&=\mean{F}.
\end{align}

The proof for the case where $f(U)^T$ forms a unitary representation follows analogous steps. Proposition\,\ref{lemma:fT_is_rep} shows that when $f(U)^T$ forms a representation of a group and $S$ is a parallel/sequential/general superchannel that transforms $k$ uses of $U$ into $f(U)$ with average fidelity $\mean{F}$, there exists a parallel/sequential/general superchannel $S'$ that transforms $k$ uses of $U$ into $f(U)$ with average fidelity $\mean{F}$ and respects
\begin{equation} 
	S'*\dketbra{U}{U}^{\otimes k} = \Big(f(U)^T_P \otimes \id_F\Big)
	 \Big( S'*\dketbra{\id}{\id}^{\otimes k}_{\bm{IO}}\Big)
	 \Big(f(U)^T_P \otimes \id_F\Big)^\dagger .
\end{equation}
Hence, for the covariant superchannel $S'$, the fidelity between $S'*\dketbra{U}{U}^{\otimes k}$ and $\dketbra{f(U)}{f(U)}$ is independent of $U$ and the average coincides with the worst-case fidelity, that is
\begin{align}
F\Big(S'*\dketbra{U}{U}^{\otimes k},\dketbra{f(U)}{f(U)}\Big)&=
\dbra{f(U)}\bigg(f(U)^T_P\otimes \id_F \Big( S'*\dketbra{\id}{\id}_{\bm{IO}}^{\otimes k} \Big) f(U)^*_P\otimes \id_F\bigg) \dket{f(U)} \nonumber \\
&=
\dbra{\id} \Big( S'*\dketbra{\id}{\id}_{\bm{IO}}^{\otimes k} \Big) \dket{\id} \nonumber \\
&=\mean{F}.
\end{align}
\end{proof}


\section{Proof of Thm.~\ref{theo:white_noise_cov} \label{app:wnc}}
	We now re-state and prove Thm.~\ref{theo:white_noise_cov} from the main text
\setcounter{theorem}{1}
\begin{theorem} 
	 Let $f:\mathcal{SU}(d)\to\mathcal{SU}(d)$ be a function respecting
\begin{itemize}
\item $f(UV)=f(U)f(V)$ for all $U,V\in \mathcal{SU}(d)$, or $f(UV)=f(V)f(U)$  $U,V\in \mathcal{SU}(d)$
\item $f(U^*)=f(U)^*$
\item For the Haar measure, the differential $\dif U$ is invariant under the substitution $\dif U\to \dif f(U)$
\end{itemize}	 
If $S$ be a parallel/sequential/general superchannel transforming $k$ uses of $U$ into $f(U)$,
there exists a parallel/sequential/general superchannel $S'$ such that
\begin{equation}
	S'* \dketbra{U}{U}^{\otimes k}_{\bm{IO}} = \eta \dketbra{f(U)}{f(U)}_{PF} + (1-\eta) \id_P\otimes \frac{\id_F}{d},
\end{equation}	 
where $\mean{F}=\eta + \frac{1-\eta}{d^2}$.
\end{theorem}
\begin{proof}
	In proposition\,\ref{prop:f_is_rep} we show that if $S$ is a parallel/sequential/general superchannel transforming $k$ uses of $U$ into $f(U)$ with average fidelity $\mean{F}$ and $f(AB)=f(A)f(B)$ for every unitary operator $A\in \SU(d)$, there exists a parallel/sequential/general superchannel $S'$ transforming $k$ uses of $U$ into $f(U)$ with average fidelity $\mean{F}$ which respects
\begin{align} \label{eq:cov_homo}
		S'*\dketbra{U}{U}^{\otimes k} 
=\Big( f(A)_P \otimes f(B)_F\Big) 
\Bigg(S' * \bigg( \dketbra{B^\dagger UA^*}{B^\dagger UA^*}^{\otimes k}\bigg)  \Bigg)
		\Big(f(A)_P \otimes f(B)_F\Big)^\dagger,
	\end{align}
for every unitary operator $A,B\in\SU(d)$.
Let us first consider the $U=\id$ case. If we set $B=A^*$, Eq.\eqref{eq:cov_homo} reads as
\begin{align} \label{eq:cov_homo2}
		S'*\dketbra{\id}{\id}^{\otimes k} = \Big( f(A)_P \otimes f(A^*)_F\Big) 
\Bigg(S' * \dketbra{\id}{\id}^{\otimes k} \Bigg)
		\Big( f(A)_P \otimes f(A^*)_F\Big) ^\dagger.
	\end{align}
Since Eq.~\ref{eq:cov_homo2} holds for every unitary operator $A$, we can then take the Haar average to obtain
\begin{align}
		S'*\dketbra{\id}{\id}^{\otimes k} =& \int \Big( f(A)_P \otimes f(A^*)_F\Big) 
\Bigg(S' * \dketbra{\id}{\id}^{\otimes k} \Bigg)
		\Big( f(A)_P \otimes f(A^*)_F\Big) ^\dagger \dif A \nonumber \\
=& \int \Big( f(A)_P \otimes f(A^*)_F\Big) 
\Bigg(S' * \dketbra{\id}{\id}^{\otimes k} \Bigg)
		\Big( f(A)_P \otimes f(A^*)_F\Big) ^\dagger \dif f(A) \nonumber \\
		=& \int \Big( f(A)_P \otimes f(A)^*_F\Big) 
\Bigg(S' * \dketbra{\id}{\id}^{\otimes k} \Bigg)
		\Big( f(A)_P \otimes f(A)^*_F\Big) ^\dagger \dif f(A) \nonumber \\
		=& \int \Big( A_P \otimes A^*_F\Big) 
\Bigg(S' * \dketbra{\id}{\id}^{\otimes k} \Bigg)
		\Big( A_P \otimes fA^*_F\Big) ^\dagger \dif A \label{eq:isotropic_twirling} \\
	&= \eta \dketbra{\id}{\id}_{PF} + (1-\eta) \frac{\id_P}{d}\otimes \id_F,
\end{align}
where the integral in equation Eq.~\eqref{eq:isotropic_twirling} is the Isotropic twirling map \cite{horodecki97reduction}, which transform any operator into a linear combination between the identity operator $\id$  and $\dketbra{\id}{\id}$.
Now, since 
\begin{equation} 
	S'*\dketbra{U}{U}^{\otimes k} = \Big(\id_P \otimes f(U)_F\Big)
	 \Big( S'*\dketbra{\id}{\id}^{\otimes k}_{\bm{IO}}\Big)
	 \Big(\id_P \otimes f(U)_F\Big)^\dagger,
\end{equation}
we have that
\begin{align} 
S'*\dketbra{U}{U}^{\otimes k} &=
\Big(\id_P \otimes f(U)_F\Big)
\Bigg( \eta \dketbra{\id}{\id}_{PF} + (1-\eta) \id_P\otimes \frac{\id_F}{d}\Bigg)
\Big(\id_P \otimes f(U)_F\Big)^\dagger \nonumber \\
&= \eta \dketbra{f(U)}{f(U)}_{PF} + (1-\eta) \id_P\otimes \frac{\id_F}{d}.
\end{align}
Now, from direct calculation (as in Eq.~\eqref{eq:white_is_fidelity}) we see that $\mean{F}=\eta + \frac{1-\eta}{d^2}$.
\end{proof}

\section{Explicit evaluation of performance operators} \label{app:explicit}
\subsection{The commutant of $U^{\otimes k}\otimes U$} \label{app:UU}
The set of operators $P$ respecting the relation
\begin{equation}
	[P,U^{\otimes k}\otimes U= 0], \quad \quad \forall U\in \SU(d)
\end{equation}
is called the commutant of $U^{\otimes k}\otimes U$ and form a linear subspace of $\L(\mathbb{C}_d^{\otimes (k+1)})$. One possible way to obtain a basis for this linear space is to characterise the projectors onto inequivalent irreducible representations of $\SU(d)$ in $U^{\otimes k}\otimes U$ and all isometries between equivalent representations. Methods to obtain such basis have been widely studied in the literature ~\cite{weylbook,harrisBook,eggeling01,alcock16}. We recommend Ref.~\cite{alcock16} for a simple and direct way to find an orthogonal basis for the commutant of $U^{\otimes k}\otimes U$ for any dimension $d$ and any $k\in\mathbb{N}$.

	A set of operators which spans the commutant of $U^{\otimes k}\otimes U$ is the set of all $(k+1)!$ permutation operators with  $(k+1)$ elements. Let $\pi:\{1,\ldots,k+1\}\to\{1,\ldots,k+1\}$ be a permutation function, a permutation operator $V_\pi \in \mathbb{C}_d^{\otimes (k+1)}$ is defined as
\begin{equation}
	V_\pi \Big(\ket{i_1}\otimes\ket{i_2}\otimes \ldots  \otimes \ket{i_{k+1}}\Big) =
	\ket{i_{\pi^{-1}(1)}}\otimes\ket{i_{\pi^{-1}(2)}}\otimes \ldots \otimes \ket{i_{\pi^{-1}(k+1)}}.
\end{equation}
	Note that permutation operators are not orthogonal, and in general, they form an ``overcomplete basis''. That is, despite spanning the space of the commutant of $U^{\otimes k}\otimes U$, the set $\{V_\pi\}_\pi$ may have linearly dependent operators. We remark however that we can always obtain an orthonormal basis from a set of operators with span a desired linear subspace, in particular, this can be done via the Gram-Schmidt orthogonalisation algorithm.

For $k=1$, a convenient orthogonal basis for the commutant of $U\otimes U$ is given by the projector onto the symmetric space and a projection onto antisymmetric space,
\begin{align}
	P^\text{sym}:= \frac{\id+F}{2}, \quad P^\text{asym}:=\frac{\id-F}{2}
\end{align}
where $F:=\sum_{ij}\ketbra{ij}{ji}$ is the flip operator.

	For $k=2$, a convenient basis for the commutant of $U\otimes U\otimes U$ for any dimension $d$ is given by~\cite{eggeling01}:
\begin{align}
	R^+:=&\frac{1}{6}\Big(  \id + V_{(12)} + V_{(23)} + V_{(31)} + V_{(123)} + V_{(321)} 	\Big), \nonumber \\
	R^-:=&\frac{1}{6}\Big(  \id - V_{(12)} - V_{(23)} - V_{(31)} + V_{(123)} + V_{(321)} 	\Big), \nonumber \\
	R^0:=&\frac{1}{3}\Big( 2\id -  V_{(123)} - V_{(321)} 	\Big), \nonumber \\
	R^1:=&\frac{1}{3}\Big(  2V_{(23)} - V_{(31)} - V_{(12)}	\Big), \nonumber \\
	R^2:=&\frac{1}{\sqrt{3}}\Big( V_{(12)} - V_{(31)} 	\Big), \nonumber \\
	R^3:=&\frac{i}{\sqrt{3}}\Big(  V_{(123)} - V_{(321)} 	\Big),
\end{align}
where the permutations $\pi$ are given in the cycle notation. The operators $R^+$, $R^-$, and $R^0$ are projectors onto the irreducible representations of $U\otimes U \otimes U$ and the operators $R^1, R^2, R^3$ are self adjoint and respect the Pauli operators relation $R^1R^2=iR^3$, $R^2R^3=iR^1$, and $R^3R^1=iR^2$ .

\subsection{The commutant of $U^{*^{\otimes k}}\otimes U$} \label{app:UstarU}
	The commutant $U^{*^{\otimes k}}\otimes U$ can also be characterised by group theoretical results. One possible way to obtain a basis for this linear space is to characterise all projectors onto inequivalent irreducible representations of $\SU(d)$ in $U^{*^{\otimes k}}\otimes U$ and all isometries between equivalent representations. Methods to obtain such basis have also been widely studied in the literature ~\cite{harrisBook,eggeling01,mozrzymas18b,alcock18}. In particular, we recommend Ref.~\cite{alcock18} for a simple and direct way to find an orthogonal basis for the commutant of  $U^{*^{\otimes k}}\otimes U$ for any dimension $d$ and any $k\in\mathbb{N}$.

	A set of operators which spans the commutant of  $U^*\otimes {U}^{\otimes k}$  is the set of all permutations $\pi$ of  $(k+1)$ elements up to a partial transposition on the first system. That is, all $(k+1)!$ operators of the form $V_\pi^{T_1}$. However, these operators are not orthonormal and are not guaranteed to be all linearly independent.
 	
	For $k=1$, a convenient orthogonal basis for the commutant of $U^*\otimes U$ is given by the projector onto the maximally entangled state and a projection onto its orthogonal complement:
\begin{align}
P^+:= \frac{\dketbra{\id}{\id}}{d} \quad P^\perp:=\id - \frac{\dketbra{\id}{\id}}{d} .
\end{align}

	For $k=2$, a convenient basis for the commutant of $U^*\otimes U^*\otimes U$  for any dimension $d$ is given by~\cite{eggeling01}:
\begin{align}
	S^+:=& \frac{\id+V}{2}\left(  \id -\frac{2X}{d+1} \right) 	\frac{\id+V}{2}, \nonumber \\
	S^-:=& \frac{\id-V}{2}\left(  \id -\frac{2X}{d-1} \right) 	\frac{\id-V}{2}, \nonumber \\
	S^0:=& \frac{1}{d^2-1}\Big(  d(X+VXV)-(XV+VX) \Big) 	, \nonumber \\
	S^1:=& \frac{1}{d^2-1}\Big(  d(XV+VX)-(X+VXV) \Big) 	, \nonumber \\
	S^2:=& \frac{1}{\sqrt{d^2-1}}\Big(  X-VXV \Big) 	, \nonumber \\
	S^3:=& \frac{i}{\sqrt{d^2-1}}\Big(  XV-VX \Big) 	, 
\end{align}
where $X:=\id_1\otimes\dketbra{\id}{\id}_{23}$ and $V:=F_{12} \otimes \id_3$ and $F:=\sum_{ij}\ketbra{ij}{ji}$ is the flip operator. The operators $S^+$, $S^-$, and $S^0$ are projectors onto the irreducible representations of $U^*\otimes U^*\otimes U$ and the operators $S^1, S^2, S^3$ are self adjoint and respect the Pauli operators relation $S^1S^2=iS^3$, $S^2S^3=iS^1$, and $S^3S^1=iS^2$. See also the Appendix B.3.6 of Ref.~\cite{bisio16} for an explicit orthonormal basis which is directly based on isometries between equivalent representations.
\subsection{Unitary complex conjugation}
	Since unitary complex conjugation is a homomorphism, the results of Appendix~\ref{app:f_is_rep} ensures that its corresponding performance operator is
\begin{equation}
	\Omega = \frac{1}{d^2}\sum_i \frac{\left(P_{\bm{I}P}^i\right)^* \otimes P^i_{\bm{O}F}}{d_i},
\end{equation}
where $\{P^i\}$ is an orthonormal basis for operators $P$ respecting $[P,U^{*^{\otimes (k+1)}}]=0$ for every $U\in\SU(d)$, or equivalently,  $[P,U^{\otimes (k+1)}]=0$ for every $U\in\SU(d)$ . We can then use the basis presented in Appendix~\ref{app:UU} to obtain explicit operators for $k=1$ and $k=2$.

For $k=1$ we have
	 \begin{equation}
	 	\Omega_{\text{conj},k=1} =
\frac{1}{d^2} \left[  \frac{ P^\text{sym}_{{I}P} \otimes P^\text{sym}_{OF}}   {(d^2+d)/2} +  
 \frac{ P^\text{asym}_{{I}P} \otimes P^\text{sym}_{{O}F}}   {(d^2-d)/2} \right].
	 \end{equation}
For $k=2$ we have	
\small
	 \begin{equation}
	 	\Omega_{\text{conj},k=2} = \frac{1}{d^2} 
\left[\frac{R^+_{\bm{I}P}\otimes R^+_{\bm{O}F}} {\tr(R^+)}
+ \frac{R^-_{\bm{I}P}\otimes R^-_{\bm{O}F}} {\tr(R^-)}
+ \frac{R^0_{\bm{I}P}\otimes R^0_{\bm{O}F}} {\tr(R^0)}
+ \frac{R^1_{\bm{I}P}\otimes R^1_{\bm{O}F}} {\tr(R^1R^1)}
+ \frac{R^2_{\bm{I}P}\otimes R^2_{\bm{O}F}} {\tr(R^2R^2)}
- \frac{R^3_{\bm{I}P}\otimes R^3_{\bm{O}F}} {\tr(R^3R^3)}
\right] .
	 \end{equation}
\normalsize

\subsection{Unitary transposition}
	Since unitary transposition is an anti-homomorphism, the results of Appendix~\ref{app:fT_is_rep} ensures that its corresponding performance operator is
\begin{equation}
	\Omega = \frac{1}{d^2}\sum_i \frac{\left(P_{\bm{O}P}^i\right)^* \otimes P^i_{\bm{I}F}}{d_i},
\end{equation}
where $\{P^i\}$ is an orthonormal basis for operators $P$ respecting $[P,U^{*^{\otimes k}}\otimes U]=0$ for every $U\in\SU(d)$. We can then use the basis presented in Appendix~\ref{app:UstarU} to obtain explicit operators for $k=1$ and $k=2$.

For $k=1$ we have
	 \begin{equation}
	 	\Omega_{\text{trans},k=1} =
\frac{1}{d^2} \left[  \frac{ \frac{\dketbra{\id}{\id}}{d}_{OP} \otimes \frac{\dketbra{\id}{\id}}{d}_{IF}}   {1} +  
\frac{ \left(\id-\frac{\dketbra{\id}{\id}}{d}\right)_{OP} \otimes \left(\id-\frac{\dketbra{\id}{\id}}{d}\right)_{{I}F} } {d-1} \right].
	 \end{equation}
For $k=2$ we have	
\small
	 \begin{equation}
	 	\Omega_{\text{trans},k=2} = \frac{1}{d^2} 
\left[\frac{S^+_{\bm{O}P}\otimes S^+_{\bm{I}F}} {\tr(S^+)}
+ \frac{S^-_{\bm{O}P}\otimes S^-_{\bm{I}F}} {\tr(S^-)}
+ \frac{S^0_{\bm{O}P}\otimes S^0_{\bm{I}F}} {\tr(S^0)}
+ \frac{S^1_{\bm{O}P}\otimes S^1_{\bm{I}F}} {\tr(S^1S^1)}
+ \frac{S^2_{\bm{O}P}\otimes S^2_{\bm{I}F}} {\tr(S^2S^2)}
- \frac{S^3_{\bm{O}P}\otimes S^3_{\bm{I}F}} {\tr(S^3S^3)}
\right] .
	 \end{equation}

\normalsize

\subsection{Unitary inversion}
	Since unitary inversion is an anti-homomorphism, the results of Appendix~\ref{app:fT_is_rep} ensures that its corresponding performance operator is
\begin{equation}
	\Omega = \frac{1}{d^2}\sum_i \frac{\left(P_{\bm{O}P}^i\right)^* \otimes P^i_{\bm{I}F}}{d_i},
\end{equation}
where $\{P^i\}$ is an orthonormal basis for operators $P$ respecting $[P,U^{\otimes (k+1)}]=0$ for every $U\in\SU(d)$. We can then use the basis presented in Appendix~\ref{app:UU} to obtain explicit operators for $k=1$ and $k=2$.

For $k=1$ we have
	 \begin{equation}
	 	\Omega_{\text{trans},k=1} =
\frac{1}{d^2} \left[  \frac{ P^\text{sym}_{OP} \otimes P^\text{sym}_{IF}}   {(d^2+d)/2} +  
 \frac{ P^\text{asym}_{OP} \otimes P^\text{sym}_{IF}}   {(d^2-d)/2} \right].
	 \end{equation}
For $k=2$ we have	
\small
	 \begin{equation}
	 	\Omega_{\text{trans},k=2} = \frac{1}{d^2} 
\left[\frac{R^+_{\bm{O}P}\otimes R^+_{\bm{I}F}} {\tr(R^+)}
+ \frac{R^-_{\bm{O}P}\otimes R^-_{\bm{I}F}} {\tr(R^-)}
+ \frac{R^0_{\bm{O}P}\otimes R^0_{\bm{I}F}} {\tr(R^0)}
+ \frac{R^1_{\bm{O}P}\otimes R^1_{\bm{I}F}} {\tr(R^1R^1)}
+ \frac{R^2_{\bm{O}P}\otimes R^2_{\bm{I}F}} {\tr(R^2R^2)}
- \frac{R^3_{\bm{O}P}\otimes R^3_{\bm{I}F}} {\tr(R^3R^3)}
\right] .
	 \end{equation}
\normalsize

\nocite{apsrev42Control} 
\bibliographystyle{0_MTQ_apsrev4-2_corrected}
\bibliography{0_MTQ_bib.bib}
\end{document}